\documentclass[11pt,letterpaper]{amsart}
\usepackage[a4paper, margin=2in, top=0.01in]{geometry}
\usepackage{amsthm,amssymb,amsmath,amstext,amsfonts}
\usepackage{enumitem,mathtools,pgfplots,pgfmath,microtype}
\usepackage{svg}
\usepackage{array}
\usetikzlibrary{arrows,shapes,calc,automata,backgrounds,decorations,petri,positioning}
\usetikzlibrary{decorations.pathreplacing,angles,quotes}
\usepackage{xcolor}
\definecolor{dgreen}{HTML}{006400}
\tikzset{axis line style/.style={thin, gray, -stealth}}
\usepackage{algorithm}
\usepackage{dirtytalk}
\usepackage{algpseudocode}

\newcommand{\Ebb}[1]{\Eb\lbb #1\rbb}
\newcommand{\Eb}{\mathbb{E}}
\newcommand{\lb}{\left(}
\newcommand{\rb}{\right)}
\newcommand{\lbb}{\left [}
\newcommand{\rbb}{\right ]}
\newcommand{\lbrb}[1]{\lb #1 \rb}
\newcommand{\D}{\mathrm{d}}
\newcommand{\so}[1]{\mathrm{o}\lbrb{#1}}
\newcommand{\bo}[1]{\mathrm{O}\lbrb{#1}}
\newcommand{\ind}[1]{\mathbbm{1}_{\left\{{#1}\right\}}}
\usepackage{bbm}
\renewcommand{\P}{\mathbb{P}}
\newcommand{\Var}{\mathbb{V}\mathrm{ar}}
\DeclareMathOperator{\sign}{sign}

\DeclareMathOperator{\totalB}{totalB}
\DeclareMathOperator{\totalD}{totalD}
\DeclareMathOperator{\totalDTrunc}{totalDTrunc}

\DeclareMathOperator{\bfsScore}{bfsScore}
\DeclareMathOperator{\dfsScore}{dfsScore}
\DeclareMathOperator{\dfsTruncScore}{dfsTruncScore}
\DeclareMathOperator{\lev}{lev}

\usepackage[utf8]{inputenc}
\usepackage[english]{babel}
\usepackage{amsthm}
\usepackage{amsmath}
\usepackage{mathtools}
\allowdisplaybreaks[3]
\usepackage{hyperref}
\usepackage{amssymb}
\usepackage{booktabs}
\usepackage{enumitem}
\usepackage{graphicx}
\usepackage{subcaption}
\usepackage{centernot}
\usepackage{tikz}
\usepackage{float}

\theoremstyle{plain}

\usepackage{setspace}
\newtheorem{theorem}{Theorem}[section]
\newtheorem{corollary}[theorem]{Corollary}
\newtheorem{definition}[theorem]{Definition}

\newtheorem{lemma}[theorem]{Lemma}

\newtheorem{question}{Question}[section]
\theoremstyle{remark}\newtheorem{remark}[theorem]{Remark}
\makeatletter

\let\originalleft\left
\let\originalright\right
\renewcommand{\left}{\mathopen{}\mathclose\bgroup\originalleft}
\renewcommand{\right}{\aftergroup\egroup\originalright}

\newcommand{\leqnomode}{\tagsleft@true\let\veqno\@@leqno}
\newcommand{\reqnomode}{\tagsleft@false\let\veqno\@@eqno}
\reqnomode

\makeatother
\calclayout

\title{BFS versus DFS for fixed-level targets in ordered trees}
\author{Stoyan Dimitrov\\
Department of Mathematics, Dartmouth college \\
\texttt{emailtostoyan@gmail.com}
 \vspace{1.5em} \\ 
Martin Minchev \\
Faculty of Mathematics and Informatics, Sofia University \\
\texttt{mjminchev@fmi.uni-sofia.bg}
 \vspace{1.5em} \\ 
Yan Zhuang \\
Department of Mathematics and Computer Science, Davidson College \\
\texttt{yazhuang@davidson.edu}
}
\begin{document}

\maketitle
\begin{abstract}
We find the average time complexity of the breadth-first search (BFS) and the depth-first search (DFS) algorithms, when one searches for a target node selected uniformly at random among all nodes at level $\ell$ in the set of ordered trees with $n$ edges. Intuition suggests that on average BFS must be asymptotically faster than DFS if and only if $\ell$, as a function of $n$, is below a certain threshold. We confirm this intuition by showing that there exists a unique constant $\lambda\approx 0.789004$, such that in expectation BFS is asymptotically faster than DFS if and only if $\ell\leq \lambda\sqrt{n}$. This gives us a practical rule to select between the two algorithms, even when we do not know the exact value of $\ell$, but only an estimate of it. Furthermore, we find the asymptotic average time complexity of BFS in the given setting for an arbitrary class of Galton--Watson trees, which includes ordered trees, binary trees, and other popular classes. We use results on the occupation measure of Brownian excursions, as well as combinatorial identities related to lattice paths. Finally, we consider the simple \textit{truncated DFS} algorithm, which can be shown easily to be asymptotically faster than both BFS and DFS when $\ell$ is known in advance. We show that in fact its asymptotic time complexity is $1/2$ of the asymptotic complexity of BFS, when $\ell = s\sqrt{n}$ for any constant $s$. Several further questions are also raised.
\end{abstract}
\section{Introduction}
Several important problems in computer science and artificial intelligence can be formulated as search problems \cite[Chapter 3]{norvig}. Some examples include the traveling salesman problem, scheduling problems, and the problem of finding optimal moves in various games. Most decision-making problems can be also naturally thought of as search problems, where the decision space is modeled as a graph or a tree. For many of these problems, a different search algorithm is optimal, depending on the instance \cite{hoos}. Some machine learning approaches exist for algorithm selection \cite{kotthoff} depending on the instance, but they are usually less useful than having an exact average-case comparison.  

In this paper, we find the average time complexity of the classical tree search algorithms---breadth-first search (BFS) and depth-first search (DFS)---when we perform a search for a random target node at a fixed level, given that we begin at the root of an ordered tree with a prescribed number of edges. We consider the most simple scenario when no additional information about the tree is available. Then, it is reasonable to assume that the target node is sampled uniformly from all nodes at the prescribed level among all trees with the given number of edges. BFS and DFS are compared in terms of their complexities for each target node level, yielding an optimal criterion for choosing between them. 

In particular, by using a folklore correspondence between ordered trees and Dyck paths, we find that the expected number of steps made by DFS, when we search for a unique target node on level $\ell$ in a tree with $n$ edges, is $\frac{\ell}{2\ell +1}(n+\ell + 1)$. Then, we find the expected number of steps made by BFS via results of Tak\'{a}cs \cite{Takacs-99} on the occupation measure of Brownian excursions. A connection between the generating function for the BFS time-complexity and the so-called \emph{Fibonacci polynomials} is also established. The average time complexities of the two algorithms are compared in the asymptotic case, when $\ell,n\to \infty$, and when $\ell$ is a function of $n$.   
Consequently, we find a unique threshold $\lambda_{n}\approx 0.789\sqrt{n}$, such that BFS is asymptotically faster than DFS if and only if $\ell\leq \lambda_n$. This confirms the intuition that BFS should have advantage when the target is close to the root, while DFS is superior when the target is far from the root. The limit of the ratio of the two complexities is also found for certain intervals for the target level $\ell$.

While ordered trees are a very general class of structures used to model various recursive processes, we also find the asymptotic average time complexity of BFS for an arbitrary class of Galton--Watson trees, which includes ordered trees, binary and $m$-ary trees, labeled trees and other classes. This generalization is obtained by utilizing some results on stochastic process and the limiting shape of trees by Aldous \cite{Aldous-98}. Finally, we propose a variation of DFS called \emph{truncated DFS}, and we show that it has a better average time performance than both BFS and DFS when the target level $\ell$ is known in advance. It is also shown that the average-case time complexity of the truncated DFS is half of those of BFS, when $\ell = s\sqrt{n}$, for any constant $s$.

Our results have several possible applications, since they are related to the complexity of the two most popular search algorithms in the universal scenario of traversing a tree. 
Note that in the cases when the size of the tree, $n$, is not known in advance, one may first use the well-known method of Knuth \cite{knuth1975estimating} to estimate that size, before using our formulas to find the better search method for different levels $\ell$.
\subsection{Problem formulation and main result}
\label{subsec:probNot}
Consider the set $\mathcal{T}$ of all unlabeled rooted ordered trees, which we simply refer to as \emph{ordered trees} or \emph{trees}. In each ordered tree, we have a fixed node called a \emph{root}, as well as a fixed ordering for the children of each node. Let $\mathcal{T}_{n}$ denote the subset of trees in $\mathcal{T}$ having exactly $n$ edges; it is well known that the number of trees in $\mathcal{T}_{n}$ is the $n$th Catalan number 
\begin{equation*}C_{n} = \frac{1}{n+1}\binom{2n}{n}.\end{equation*}
Note that each tree in $\mathcal{T}_{n}$ has $n+1$ nodes. The \emph{level} of a node $v$ is the length of the unique path from $v$ to the root; thus the root is at level $0$. Trees will be drawn with the root at the top.

For trees, BFS explores the entire neighborhood of the current node and then continues its execution recursively with the children of the current node one-by-one, while DFS follows one path for as long as possible and backtracks when stuck. DFS is equivalent to a pre-order traversal of the tree, that is, the root is visited first and then recursively the subtrees from left to right. Figure \ref{fig:bfs_dfs_order} illustrates the orders in which BFS and DFS will explore the nodes of a particular tree, if the initial node is the root.
\begin{figure}[ht!]
\centering
\begin{subfigure}{.48\textwidth}
  \centering
    \begin{tikzpicture}  
  [baseline=0, scale=0.8,auto=center]
    
  \node[circle, draw, fill=black!100, inner sep=0pt, minimum width=4pt] (a1) at (0,10) {};
  \node[above] at (0,10) {0}; 
  \node[circle, draw, fill=black!100, inner sep=0pt, minimum width=4pt] (a2) at (-1.5,9) {};
  \node[left] at (-1.5,9) {1};
  \node[circle, draw, fill=black!100, inner sep=0pt, minimum width=4pt] (a3) at (0,9) {};
  \node[left] at (0,9) {2};
  \node[circle, draw, fill=black!100, inner sep=0pt, minimum width=4pt] (a4) at (1.5,9) {};
  \node[left] at (1.5,9) {3};
  \node[circle, draw, fill=black!100, inner sep=0pt, minimum width=4pt] (a5) at (-2.5,8) {}; 
  \node[left] at (-2.5,8) {4};
  \node[circle, draw, fill=black!100, inner sep=0pt, minimum width=4pt] (a6) at (-0.5,8) {}; 
  \node[left] at (-0.5,8) {5};
  \node[circle, draw, fill=black!100, inner sep=0pt, minimum width=4pt] (a7) at (-2.5,7) {}; 
  \node[left] at (-2.5,7) {7};
  \node[circle, draw, fill=black!100, inner sep=0pt, minimum width=4pt] (a8) at (1.5,8) {};
  \node[left] at (1.5,8) {6};
  \node[circle, draw, fill=black!100, inner sep=0pt, minimum width=4pt] (a9) at (0.5,7) {};
  \node[left] at (0.5,7) {8};
  \node[circle, draw, fill=black!100, inner sep=0pt, minimum width=4pt] (a10) at (2.5,7) {};
  \node[left] at (2.5,7) {9};
  
  \draw (a1) -- (a2);
  \draw (a1) -- (a3);  
  \draw (a1) -- (a4);  
  \draw (a2) -- (a5);  
  \draw (a2) -- (a6);
  \draw (a5) -- (a7);  
  \draw (a4) -- (a8);  
  \draw (a8) -- (a9);  
  \draw (a8) -- (a10);  
\end{tikzpicture}
\vspace{-5cm}
  \caption{The order of exploration for the BFS algorithm.}
\end{subfigure}\hspace{5mm}
\begin{subfigure}{.48\textwidth}
  \centering
      \begin{tikzpicture}  
  [baseline=0, scale=0.8,auto=center]
    
  \node[circle, draw, fill=black!100, inner sep=0pt, minimum width=4pt] (a1) at (0,10) {};
  \node[above] at (0,10) {0}; 
  \node[circle, draw, fill=black!100, inner sep=0pt, minimum width=4pt] (a2) at (-1.5,9) {};
  \node[left] at (-1.5,9) {1};
  \node[circle, draw, fill=black!100, inner sep=0pt, minimum width=4pt] (a3) at (0,9) {};
  \node[left] at (0,9) {5};
  \node[circle, draw, fill=black!100, inner sep=0pt, minimum width=4pt] (a4) at (1.5,9) {};
  \node[left] at (1.5,9) {6};
  \node[circle, draw, fill=black!100, inner sep=0pt, minimum width=4pt] (a5) at (-2.5,8) {}; 
  \node[left] at (-2.5,8) {2};
  \node[circle, draw, fill=black!100, inner sep=0pt, minimum width=4pt] (a6) at (-0.5,8) {}; 
  \node[left] at (-0.5,8) {4};
  \node[circle, draw, fill=black!100, inner sep=0pt, minimum width=4pt] (a7) at (-2.5,7) {}; 
  \node[left] at (-2.5,7) {3};
  \node[circle, draw, fill=black!100, inner sep=0pt, minimum width=4pt] (a8) at (1.5,8) {};
  \node[left] at (1.5,8) {7};
  \node[circle, draw, fill=black!100, inner sep=0pt, minimum width=4pt] (a9) at (0.5,7) {};
  \node[left] at (0.5,7) {8};
  \node[circle, draw, fill=black!100, inner sep=0pt, minimum width=4pt] (a10) at (2.5,7) {};
  \node[left] at (2.5,7) {9};
  
  \draw (a1) -- (a2);
  \draw (a1) -- (a3);  
  \draw (a1) -- (a4);  
  \draw (a2) -- (a5);  
  \draw (a2) -- (a6);
  \draw (a5) -- (a7);  
  \draw (a4) -- (a8);  
  \draw (a8) -- (a9);  
  \draw (a8) -- (a10);  
\end{tikzpicture}
\vspace{-5cm}
  \caption{The order of exploration for the DFS algorithm.}
\end{subfigure}
\caption{The $\bfsScore$ and $\dfsScore$ of each node in a tree with $9$ edges.}
\label{fig:bfs_dfs_order}
\end{figure}
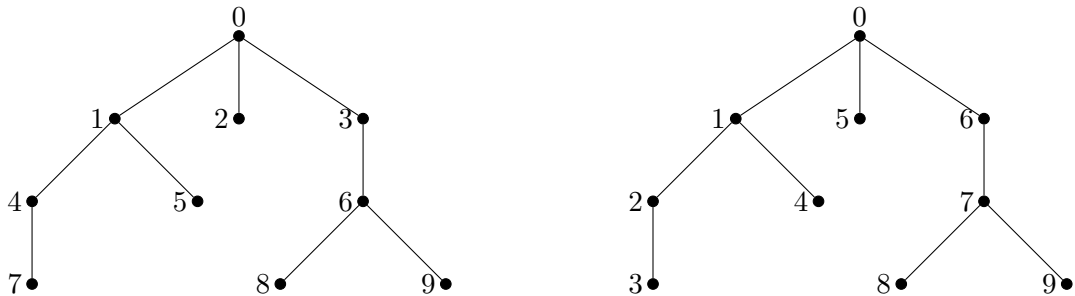
\newpage
\noindent Let us call the label received by a node $v$, when doing the traversals, the \emph{bfsScore} and the \emph{dfsScore} of $v$, respectively. We define them formally as follows: 
\begin{align*}
\bfsScore (v) &\coloneqq \text{the number of nodes visited before $v$ when using BFS}. \\
\dfsScore (v) &\coloneqq \text{the number of nodes visited before $v$ when using DFS}.
\end{align*}
It is important to note that $\bfsScore(v)$ and $\dfsScore(v)$ also give the number of steps taken in searching for $v$ when one starts at the root and uses BFS and DFS, respectively. Thus, these two quantities measure the time complexity of the two algorithms when the target is $v$.

Our goal will be to compare the average time complexity of BFS and DFS---i.e., the average $\bfsScore$ and the average $\dfsScore$---under a minimal set of assumptions for the search. In particular, consider a search from the root of some tree $T\in \mathcal{T}_{n}$ to a specific target node $x$ in $T$, by using BFS or DFS. We assume that we have a function telling us whether the current node is the target $x$, but we do not know the number of children of the current node or anything else about it. If the target $x$ is chosen uniformly at random among all the nodes in the trees in $\mathcal{T}_{n}$, then it is readily seen that the expected number of steps needed to reach $x$ is the same for both BFS and DFS, namely 
$$\frac{|\mathcal{T}_{n}|\sum_{i=1}^{n}i}{n|\mathcal{T}_{n}|} = \frac{1}{n}\frac{n(n+1)}{2} = \frac{n+1}{2}.$$ Would this change if the level of the target node is fixed? We compare the expected number of steps of the two algorithms, if $x$ is selected uniformly at random from all nodes at a fixed level $\ell$ among the trees in $\mathcal{T}_{n}$? For instance, when $n=3$ and $\ell=2$, we have $5$ such nodes as seen in Figure~\ref{fig:n3l2}. Note that under this setting, the chance for $T$ to be each of the trees in $\mathcal{T}_{n}$ is not the same, but it is proportional to the number of nodes at level $\ell$.
\begin{figure}[ht!]
    \centering
\scalebox{0.9}{
     \begin{tikzpicture}  
\begin{scope}
\node[circle, draw, fill=black!100, inner sep=0pt, minimum width=4pt] (a1) at (-5,10) {};
\node[circle, draw, fill=black!100, inner sep=0pt, minimum width=4pt] (a2) at (-5,9) {};
\node[circle, draw, fill=black!100, inner sep=0pt, minimum width=4pt] (a3) at (-5,8) {};
\draw (-5,8) [fill] circle (5pt) [fill=none];
\node[circle, draw, fill=black!100, inner sep=0pt, minimum width=4pt] (a4) at (-5,7) {};
  
  \draw (a1) -- (a2);
  \draw (a2) -- (a3);
  \draw (a3) -- (a4);
\end{scope}

\begin{scope}[xshift=4cm]
\node[circle, draw, fill=black!100, inner sep=0pt, minimum width=4pt] (a1) at (-5,10) {};
\node[circle, draw, fill=black!100, inner sep=0pt, minimum width=4pt] (a2) at (-5,9) {};
\node[circle, draw, fill=black!100, inner sep=0pt, minimum width=4pt] (a3) at (-6,8) {};
\draw (-6,8) [fill] circle (5pt) [fill=none];
\node[circle, draw, fill=black!100, inner sep=0pt, minimum width=4pt] (a4) at (-4,8) {};
\draw (-4,8) [fill] circle (5pt) [fill=none];
  
  \draw (a1) -- (a2);
  \draw (a2) -- (a3);
  \draw (a2) -- (a4);
\end{scope}

\begin{scope}[xshift=8cm]
\node[circle, draw, fill=black!100, inner sep=0pt, minimum width=4pt] (a1) at (-5,10) {};
\node[circle, draw, fill=black!100, inner sep=0pt, minimum width=4pt] (a2) at (-6,9) {};
\node[circle, draw, fill=black!100, inner sep=0pt, minimum width=4pt] (a3) at (-4,9) {};
\node[circle, draw, fill=black!100, inner sep=0pt, minimum width=4pt] (a4) at (-6,8) {};
\draw (-6,8) [fill] circle (5pt) [fill=none];
  
  \draw (a1) -- (a2);
  \draw (a1) -- (a3);
  \draw (a2) -- (a4);
\end{scope}

\begin{scope}[xshift=12cm]
\node[circle, draw, fill=black!100, inner sep=0pt, minimum width=4pt] (a1) at (-5,10) {};
\node[circle, draw, fill=black!100, inner sep=0pt, minimum width=4pt] (a2) at (-6,9) {};
\node[circle, draw, fill=black!100, inner sep=0pt, minimum width=4pt] (a3) at (-4,9) {};
\node[circle, draw, fill=black!100, inner sep=0pt, minimum width=4pt] (a4) at (-4,8) {};
\draw (-4,8) [fill] circle (5pt) [fill=none];
  
  \draw (a1) -- (a2);
  \draw (a1) -- (a3);
  \draw (a3) -- (a4);
\end{scope}

\begin{scope}[xshift=16cm]
\node[circle, draw, fill=black!100, inner sep=0pt, minimum width=4pt] (a1) at (-5,10) {};
\node[circle, draw, fill=black!100, inner sep=0pt, minimum width=4pt] (a2) at (-6,9) {};
\node[circle, draw, fill=black!100, inner sep=0pt, minimum width=4pt] (a3) at (-5,9) {};
\node[circle, draw, fill=black!100, inner sep=0pt, minimum width=4pt] (a4) at (-4,9) {};
  
  \draw (a1) -- (a2);
  \draw (a1) -- (a3);
  \draw (a1) -- (a4);
\end{scope}
\end{tikzpicture}
}
    \caption{The target node is selected uniformly at random among the five circled nodes at level $\ell=2$, when we have $n=3$ edges.}
    \label{fig:n3l2}
\end{figure}
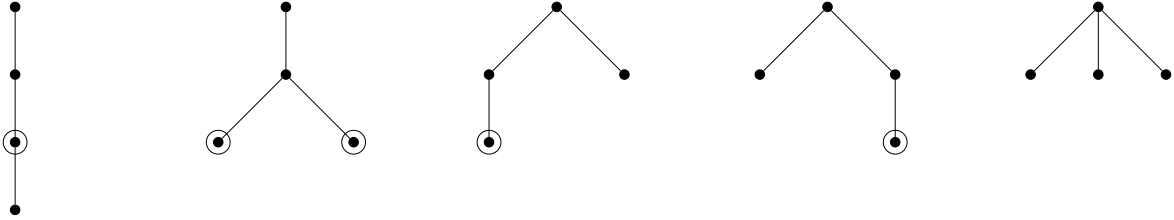

\noindent Intuition suggests that BFS is a better choice if $\ell$ is relatively small, while DFS should be preferred when $\ell$ is large. What is the exact place where DFS becomes faster, and is this place unique? To answer this question, we need to compute and compare the two expectations given below, where $\lev(T,\ell)$ denotes the set of nodes at level $\ell$ in the tree $T$:
\begin{equation*}
\mathop{\mathbb{E}}_{\substack{x\in \lev(T,\ell) \\T\in \mathcal{T}_{n}}}(\bfsScore (x)) \quad\text{and}\quad \mathop{\mathbb{E}}_{\substack{x\in \lev(T,\ell) \\T\in \mathcal{T}_{n}}}(\dfsScore (x)).
\end{equation*}
Each of the two expectations is a fraction with denominator equal to the total number of nodes at level $\ell$ among the trees in $\mathcal{T}_{n}$. A result of Dershowitz and Zaks gives us this denominator.
\begin{theorem}[Dershowitz--Zaks {\cite[Corollary 5.3]{dersh}}] \label{t-dersh}
The total number of nodes at level $\ell$ among trees in $\mathcal{T}_{n}$ is equal to
\begin{equation}
\label{eq:N}
\frac{2\ell +1}{2n+1}\binom{2n+1}{n-\ell }.
\end{equation}
\end{theorem}
Therefore, comparing the average time complexity of BFS and DFS is equivalent to comparing the two sums defined below:
\begin{equation*}
\totalB(n,\ell ) \coloneqq\sum_{T\in{\mathcal T}_{n}} \sum\limits_{v\in \lev(T,\ell )} \bfsScore(v) \quad \text{and} \quad
\totalD(n,\ell ) \coloneqq\sum_{T\in{\mathcal T}_{n}} \sum\limits_{v\in \lev(T,\ell )} \dfsScore(v).
\end{equation*}
The main result of this paper is the following theorem.
\begin{theorem}
\label{th:threshold}
For each $n\geq 1$, there is a threshold $\lambda_n$ such
that $\totalB(n,\ell ) \leq \totalD(n,\ell )$
if and only if $\ell \leq \lambda_n$. Moreover,
as $n \to \infty$, we have
\[
\lambda_n \sim \lambda\sqrt{n}
\]
where $\lambda \approx 0.789004$ is the unique positive root of the
equation
\[
xe^{-x^2} - 2\sqrt{\pi}\lbrb{\Phi\lbrb{2x\sqrt{2}}-
\Phi\lbrb{x\sqrt{2}}} = 0
\]
and $\Phi$ is the cumulative distribution function of the standard normal distribution.
\end{theorem}
Therefore, BFS is asymptotically faster than DFS if and only if $\ell\leq \lambda\sqrt{n}$, for the given constant $\lambda$. We also give explicit formulas for $\totalB(n,\ell )$ and $\totalD(n,\ell )$, and thus for the average time complexities of BFS and DFS in the considered settings. 

\subsection{Related works}
\label{subsec:relWorks}
The average performance of search algorithms over different families of graphs is a direction in which not many exact results are known. An article of Everitt and Hutter \cite{everitt} contains the only study we found considering the average-case complexities of BFS and DFS on trees for a fixed target level. However, their results are about complete binary trees with given depth as opposed to ordered trees with arbitrary depth, as in our case. Furthermore, the way they sample the target nodes is different and the obtained expressions for the expected running times, in case of a single target node, are only approximations (see \cite[Section 4]{everitt}). Another article by Stone and Sipala \cite{stone} investigates the average time complexity of two variations of DFS in binary trees. 

If the number of the vertices in the tree is not known in advance, one can use the simple algorithm of Knuth \cite{knuth1975estimating} to obtain an unbiased estimator of this number, before using our results. Knuth's method can be applied when having a tree in any given family of Galton--Watson trees. Some alternative methods for estimation of the size of the tree were suggested by Kilby et al. \cite{kilby2006estimating}. Two useful articles on algorithm selection for combinatorial search problems are the work of Rice \cite{rice1976algorithm} and the survey by Kotthoff \cite{kotthoff}.

Several other results exist on both BFS and DFS, and their variations. The average performance of the $A^{*}$ algorithm (which reduces to BFS if no heuristic information is available) is studied in \cite{sen2002average}, when searching in graphs. A recent article by Jacquet and Janson \cite{jacjanson} gives asymptotic results concerning several statistics when DFS begins at a random node of a random directed graph with certain distribution of the outdegrees for its vertices. In addition, DFS has been used to obtain some classical results in enumerative combinatorics \cite{ges1, ges2}. 
\subsection{Outline of the paper}
\label{sec:summary}
The organization of our paper is as follows:
\begin{itemize}[leftmargin=25bp] \setlength\itemsep{5bp}

\item In Section~\ref{sec:totalD}, we find a simple exact formula for $\totalD(n,\ell )$. Our proof for this formula utilizes a classical bijection between ordered trees and Dyck paths; this allows us to translate the problem to establishing an identity, for which we give a combinatorial proof by interpreting both sides as counting certain lattice paths. 

\item In Section~\ref{sec:totalB}, we obtain a summation formula for $\totalB(n,\ell )$ via results of Tak\'{a}cs \cite{Takacs-99} used to study the occupation measure of Brownian excursions. In addition, we find the asymptotic behavior of $\totalB(n,\ell )$ depending on $\ell$. 

\item Section~\ref{sec:threshold} is devoted to the proof of our main result---Theorem~\ref{th:threshold}---which establishes a unique threshold for which BFS becomes asymptotically faster than DFS. In Section~\ref{subsec:ratios}, we show that when $n^{1/3}\prec\ell\prec\sqrt n$, then BFS can be arbitrarily faster than DFS, while if $\sqrt{n}\prec \ell \prec \sqrt{n\ln{n}}$, then DFS can be at most twice as fast. Here and throughout the paper, $f(x) \prec g(x)$ means $f(x) = o(g(x))$.

\item In Section~\ref{sec:fibo}, we find a different formula for $\totalB(n,\ell )$ by using generating function techniques. This approach reveals a curious identity expressing the generating function for $\totalB(n,\ell )$ in terms of the Catalan generating function and Fibonacci polynomials; see Theorem~\ref{t-gf}~(b). We then derive an alternative formula for $\totalB(n,\ell )$ by extracting coefficients from the generating function formula we found; see Theorem~\ref{t-totalBsum}. Along the way---in Section \ref{subsec:lemmas}---we give several identities related to Fibonacci polynomials needed for proving Theorem~\ref{t-gf}. 

\item Section~\ref{sec: CRT} establishes an expression for the asymptotic value of $\totalB(n,\ell)$, for any given $\ell$, when we have an arbitrary Galton--Watson tree; see Equation~\eqref{eq:totalB-GW-asymp}. 

\item In Section~\ref{sec:truncDFS}, we introduce a new algorithm called \emph{truncated DFS}, which is faster in expectation than both BFS or DFS when we know the level $\ell$ in advance. We show that this algorithm is twice as fast as BFS, when $\ell = s\sqrt{n}$, for any constant $s$. 

\item We conclude in Section~\ref{sec:q} by presenting several questions for future study.
\end{itemize}

\section{Exact formula for the average time complexity of DFS}
\label{sec:totalD}
In this section, we establish a surprisingly simple formula for $\totalD(n,\ell )$ by using a classical correspondence between ordered trees and lattice paths. When we refer to a \emph{lattice path} from $(a,b)$ to $(c,d)$, we will mean a path in $\mathbb{Z}^{2}$ starting at $(a,b)$, ending at $(c,d)$, and consisting of a sequence of steps from a prescribed \emph{step set} $S$. We will work with two kinds of lattice paths, defined as follows:
\begin{itemize}[leftmargin=25bp] \setlength\itemsep{5bp}
\item A \emph{diagonal path} is a lattice path with step set $S = \{(0,1), (1,0)\}$, with $(0,1)$ and $(1,0)$ called $N$ and $E$ steps, respectively (from North and East).
\item A \emph{horizontal path} is a lattice path with step set $S = \{(1,1), (1,-1)\}$, with $(1,1)$ and $(1,-1)$ called $U$ and $D$ steps, respectively (from Up and Down).
\end{itemize}
A \emph{Dyck path} of \emph{semilength} $n$ is a diagonal path from $(0,0)$ to $(n,n)$ that stays weakly above the diagonal $y=x$. Note that a simple bijection between horizontal and diagonal paths is to send a horizontal $U$ (respectively $D$) step to a diagonal $N$ (respectively $E$) step. Note also that the set of Dyck paths maps to the set of horizontal paths from $(0,0)$ to $(2n,0)$, staying weakly above $y=0$. 

Given $a\leq b$ and $c \leq d$, let $[(a,b)\rightarrow(c,d)]$ denotes the number of diagonal paths from $(a,b)$ to $(c,d)$ that stay weakly above the diagonal $y=x$. We shall need the following well-known fact from lattice path enumeration.

\begin{lemma}[{\cite[Corollary 10.3.2]{krat}}]
    \label{lemma:paths}
    For every $j\geq i$, the number of diagonal paths from $(0,0)$ to $(i,j)$ that stay weakly above $y=x$ is given by
    \begin{equation*}
        [(0,0)\rightarrow(i,j)] = \frac{j+1-i}{j+i+1}\binom{j+i+1}{i}.
    \end{equation*} 
\end{lemma}
\noindent For example, when $i=n$ and $j=n$, we get that the number of Dyck paths of semilength $n$ is 
\begin{equation*}
[(0,0)\rightarrow(n,n)] = \frac{1}{2n+1}\binom{2n+1}{n} = \frac{1}{n+1}\binom{2n}{n} = C_n,
\end{equation*}
the $n$th Catalan number. 


Below is our main result of this section.

\begin{theorem}
\label{th:totalD}
For every $n\geq 0$ and $0 \leq \ell \leq n$, we have
    \begin{equation}
    \label{eq:totalD}
        \totalD(n,\ell ) = \ell\binom{2n}{n-\ell }.
    \end{equation}
\end{theorem}
\begin{proof}
We will use a folklore bijection between ordered trees and Dyck paths, known in the literature as the \emph{glove} or the \emph{accordion} bijection \cite{callan}. Execute a pre-order traversal for an ordered tree (i.e., run DFS without a target). During this process, write a $N$ step whenever we go further away from the root, and an $E$ step every time we go closer to the root. Whenever we are at a given node $v$ on level $\ell$ in an ordered tree, we will be on the line $y=x+\ell$ in the corresponding Dyck path since we have $\ell$ more $N$ steps than $E$ steps and thus $v$ has coordinates $(i,i+\ell )$ for some $i\geq 0$. It is readily seen that $\dfsScore(v)$ is equal to the number of $N$ steps in the prefix of the corresponding Dyck path that ends at the first point on the line $y=x+\ell $ corresponding to $v$. Note that the last step in this prefix is always an $N$ and the last step in all other prefixes ending at points corresponding to $v$ is an $E$ since we go to $v$ from a node at level $\ell+1$, when we visit it for the second time or more. 

It follows that the contribution of a given tree $T$ to $\totalD(n,\ell )$ is the sum of the number of $N$ steps in all prefixes $p$ of the Dyck path corresponding to $T$, such that $p$ ends with an $N$ step at the line $y=x+\ell $. Therefore, it suffices to count pairs $(\mu,\nu)$ of diagonal paths---both staying weakly above the diagonal $y=x$---such that $\mu$ begins at $(0,0)$ and ends with an $N$ step at $y=x+\ell$, $\nu$ begins at the endpoint of $\mu$ and ends at the point $(n,n)$, and the pair $(\mu,\nu)$ is weighted by the number of $N$ steps in $\mu$ (see Figure~\ref{fig:pairsDP} below). 

 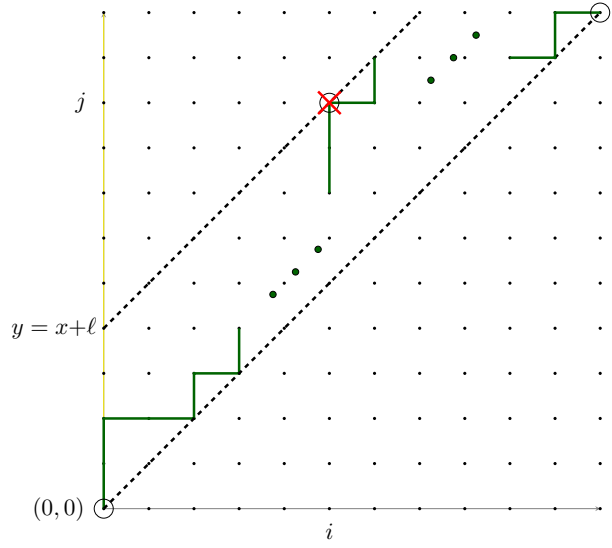
\begin{figure}[!ht]
    \centering
    \scalebox{0.5}{
\begin{tikzpicture}
\draw[axis line style] (0,0)-- (0,11);
\draw[axis line style] (0,0)-- (11,0);
\draw[yellow] (0,0) grid (0,11);
\foreach \i in {0,...,11} 
    \foreach \j in {0,...,11}
        \fill (\i,\j) circle (1pt);
\draw[ultra thick,dgreen,postaction=decorate] (0,0) -- (0,2) -- (2,2) -- (2,3) -- (3,3) -- (3,4);
\draw[ultra thick,dgreen,postaction=decorate] (5,7) -- (5,8);
\draw[ultra thick,dashed, black,postaction=decorate] (0,0) -- (11,11);
\draw[ultra thick,dashed, black,postaction=decorate] (0,4) -- (7,11);
\draw[fill=dgreen] (3.75,4.75) circle (2pt);
\draw[fill=dgreen] (4.25,5.25) circle (2pt);
\draw[fill=dgreen] (4.75,5.75) circle (2pt);
\draw[ultra thick,dgreen,postaction=decorate] (5,8) -- (5,9) -- (6,9) -- (6,10) ;
\node[scale=1.35,align=left] at (-1,0) {$(0,0)$};
\draw[fill=dgreen] (7.25,9.5) circle (2pt);
\draw[fill=dgreen] (7.75,10) circle (2pt);
\draw[fill=dgreen] (8.25,10.5) circle (2pt);
\draw[ultra thick,dgreen,postaction=decorate] (9,10) -- (10,10) -- (10,11) -- (11,11);
\node[scale=1.35,align=left] at (-1.1,4) {$y=x{+}\ell$};
\node[scale=1.35] at (5,-0.5) {$i$};
\node[scale=1.35,align=left] at (-0.5,9) {$j$};
\draw[fill=none] (0,0) circle (6pt);
\draw[fill=none] (5,9) circle (6pt);
\draw[fill=none] (11,11) circle (6pt);
\draw[ultra thick, red] (4.75,8.75)-- (5.25,9.25);
\draw[ultra thick, red] (5.25,8.75)-- (4.75,9.25);
\end{tikzpicture}}
\caption{A Dyck path that is an image of an ordered tree under the folklore bijection. The $\dfsScore$ of the tree node corresponding to the point marked with a cross is equal to the number of $N$ steps in the prefix of the path ending at that point.}
\label{fig:pairsDP}
\end{figure}
If $\mu$ has $j$ $N$ steps and $i$ $E$ steps, then its endpoint is $(i,j)$ and $j = i + \ell$. Hence $j \geq \ell$ and we have
\begin{align*}
\totalD(n,\ell ) &= \sum\limits_{j = l}^{n} j\cdot[(0,0)\rightarrow(i,j-1)]\cdot[(i,j)\rightarrow(n,n)] \\
&= \sum\limits_{j = \ell }^{n} j\cdot[(0,0)\rightarrow(j-\ell ,j-1)]\cdot[(j-\ell ,j)\rightarrow(n,n)].
\end{align*}
By symmetry and Lemma~\ref{lemma:paths}, we have
\begin{align*}
\totalD(n,\ell ) &= \sum\limits_{j = \ell }^{n} j\cdot[(0,0)\rightarrow(j-\ell ,j-1)]\cdot[(0,0)\rightarrow(n-j,n-j+\ell )] \\
&= \sum\limits_{j = \ell }^{n} \frac{j\ell }{2j-\ell } \binom{2j-\ell }{j} \frac{\ell + 1}{2n-2j+\ell +1} \binom{2n-2j+\ell +1}{n-j} \\
&= \ell \sum\limits_{j = \ell }^{n} \frac{\ell +1}{2n-2j+\ell +1} \binom{2j-\ell -1}{j-1}\binom{2n-2j+\ell +1}{n-j}.
\end{align*}
It remains to prove that 
\begin{equation}
\label{eq:paths}
    \sum\limits_{j = \ell }^{n} \frac{\ell +1}{2n-2j+ \ell +1} \binom{2j- \ell -1}{j-1}\binom{2n-2j+\ell +1}{n-j} = \binom{2n}{n- \ell }.
\end{equation}
The last follows  from Gould's generalized Vandermonde convolution~\cite[Eq.~(11)]{Gould-56}, with the substitution $N=n-\ell$, $k=n-j$, $\alpha=\ell-1$, $\gamma=\ell+1$, and $\beta=2$.

However, we provide also a combinatorial proof by showing that both sides of \eqref{eq:paths} count horizontal paths from $(0,0)$ to $(2n,2)$, i.e., those with $n+1$ $U$s and $n-1$ $D$s crossing the horizontal line $y=\ell $.

Let $P$ be such a path. We decompose $P$ as $P = P'P''$, where $P'$ ends at the first point at which $P$ reaches $y=\ell +1$, and where $P''$ is the remaining path (see Figure~\ref{fig:p'p''} below). Note that the number of $D$ steps of $P''$ can be from $\ell-1$ to $n-1$. If $P''$ has $j-1$ $D$ steps, then it must have $(j-1)-(\ell -1)$ $U$ steps, and $P''$ can be any path of this kind. Thus, there are $\binom{2(j-1)-(\ell -1)}{j-1} = \binom{2j-\ell -1}{j-1}$ possibilities for $P''$. 

As for $P'$, first observe that since we have $2n$ steps in $P$ and $2j-\ell -1$ steps in $P''$, there are $2n-2j+\ell +1$ steps in $P'$. In particular, there are $n-j$ $D$ steps in $P'$, as we have $n-1$ $D$ steps in $P$ and $j-1$ such steps in $P''$. There are $\binom{2n-2j+\ell +1}{n-j}$ paths with $n-j$ $D$ steps and $2n-2j+\ell +1$ steps in total, but we only want to count those that do not cross $y=\ell $. Lemma~\ref{lemma:paths} essentially gives what we want, although we need to translate our situation to the setting of diagonal paths. 

Upon excluding the last step in $P'$ (which must be an $U$ step ending at $y=\ell +1$), we must count the horizontal paths starting at $(0,0)$ that consist of $n-j$ $D$ steps, $n-j+\ell $ $U$ steps, and which do not cross $y=\ell $. Mapping any such horizontal path to a diagonal path via the canonical bijection gives a diagonal path with the same number of $N$ and $E$ steps and which stays weakly above $y=x$. By Lemma~\ref{lemma:paths}, the number of these paths is exactly $\binom{2n-2j+\ell +1}{n-j}\frac{\ell +1}{2n-2j+\ell +1}$, so this counts the possibilities for $P'$. Therefore, the left-hand side of \eqref{eq:paths} counts the paths $P$ satisfying the prescribed conditions.
 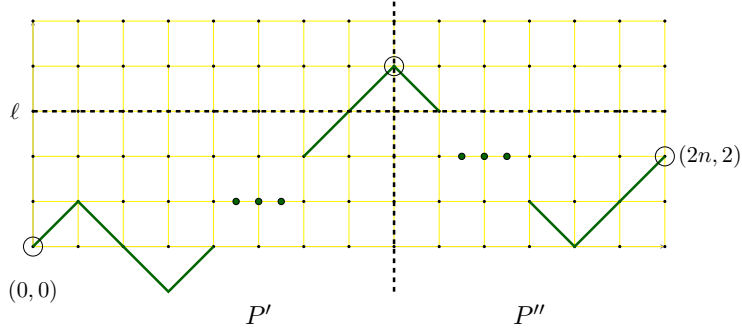
\begin{figure}[!ht]
 \vspace{-5bp}
    \centering
    \scalebox{0.5}{
\begin{tikzpicture}
\draw[axis line style] (0,0)-- (14,0);
\draw[axis line style] (0,0)-- (0,5);
\draw[yellow] (0,0) grid (14,5);
\foreach \i in {0,...,14}
    \foreach \j in {0,...,5}
        \fill (\i,\j) circle (1pt);
\draw[ultra thick,dgreen,postaction=decorate] (0,0) -- (1,1) -- (3,-1) -- (4,0);
\draw[ultra thick,dgreen,postaction=decorate] (6,2) -- (7,3) -- (8,4) -- (9,3);
\draw[ultra thick,dashed, black,postaction=decorate] (0,3) -- (14,3);
\draw[ultra thick,dashed, black,postaction=decorate] (8,-1) -- (8,5.5);
\draw[fill=dgreen] (9.5,2) circle (2pt);
\draw[fill=dgreen] (10,2) circle (2pt);
\draw[fill=dgreen] (10.5,2) circle (2pt);
\draw[ultra thick,dgreen,postaction=decorate] (11,1) -- (12,0) -- (13,1) -- (14,2);
\draw[fill=dgreen] (4.5,1) circle (2pt);
\draw[fill=dgreen] (5,1) circle (2pt);
\draw[fill=dgreen] (5.5,1) circle (2pt);
\node[scale=1.3] at (0,-1) {$(0,0)$};
\node[scale=1.3,align=right] at (15,2) {$(2n,2)$};
\node[scale=1.3,align=left] at (-0.4,3) {$\ell$};
\node[scale=1.5,align=left] at (5,-1.5) {$P'$};
\node[scale=1.5,align=left] at (11,-1.5) {$P''$};
\draw[fill=none] (0,0) circle (6pt);
\draw[fill=none] (8,4) circle (6pt);
\draw[fill=none] (14,2) circle (6pt);
\end{tikzpicture}}
\caption{The decomposition $P=P'P''$ of the paths crossing $y=\ell$ counted by the two sides of Equation~\eqref{eq:paths}.}
\label{fig:p'p''}
\end{figure}
On the other hand, we can use the \emph{reflection principle} \cite{hump} to show that the number of horizontal paths from $(0,0)$ to $(2n,2)$ that cross $y=\ell$ is exactly $\binom{2n}{n-\ell }$. Taking the same decomposition $P=P'P''$, the reflection of $P'$ about $y=\ell +1$ together with $P''$ gives us a horizontal path starting at $(0,2\ell +2)$ and ending at $(2n,2)$. Each of these paths have $n+\ell $ $D$ steps and $n-\ell $ $U$ steps. The described map is surjective since every horizontal path from $(0,2\ell +2)$ to $(2n,2)$ crosses $y=\ell +1$. Given such a path, we can take the first point where it intersects $y=\ell +1$ and reflect the prefix ending there, yielding a horizontal path from $(0,0)$ to $(2n,2)$ that crosses $y=\ell $.
\end{proof}

Using Theorems \ref{t-dersh} and \ref{th:totalD}, we can now obtain the expected $\dfsScore$ over all nodes at a given level $\ell$ among the trees in $\mathcal{T}_{n}$: 
$$
\mathop{\mathbb{E}}_{\substack{x\in \lev(T,\ell ) \\T\in \mathcal{T}_{n}}}(\dfsScore(x)) = \frac{\ell\binom{2n}{n-\ell }}{\frac{2\ell +1}{2n+1}\binom{2n+1}{n-\ell }} = \frac{\ell }{2\ell +1}(n+\ell +1).
$$
When $\ell \to \infty$ and $n \to \infty$, with $l\prec n$, this expectation is approximately $n/2$.

\section{Exact formula for the average time complexity of BFS}
\label{sec:totalB}

In this section, we give an exact summation formula for $\totalB(n,\ell )$. To begin, let us define some random variables on the trees in $\mathcal{T}_{n}$:
\begin{equation}\label{def:Random variables}
    \begin{split}
        h_n(\ell ) &\coloneqq \text{total number of vertices at level $\ell$},\\
H_n(\ell ) &\coloneqq \text{total number of vertices at levels $\leq \ell $},\\
\nu_n(\ell ) &\coloneqq \text{total number of vertices at levels $\geq \ell $},\\
S_n(\ell ) &\coloneqq \text{sum of bfsScores on level $\ell$}. 
    \end{split}
\end{equation}
For a sample realization of these random variables, see Figure \ref{fig:bfs_calculation}.
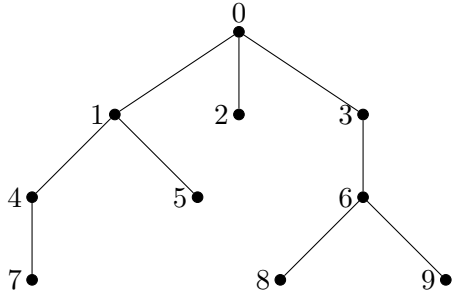
\begin{figure}[ht!]
\centering
\begin{subfigure}{.5\textwidth}
  \centering
    \begin{tikzpicture}  
  [baseline=0, scale=1.1,auto=center]
    
  \node[circle, draw, fill=black!100, inner sep=0pt, minimum width=4pt] (a1) at (0,10) {};
  \node[above] at (0,10) {0}; 
  \node[circle, draw, fill=black!100, inner sep=0pt, minimum width=4pt] (a2) at (-1.5,9) {};
  \node[left] at (-1.5,9) {1};
  \node[circle, draw, fill=black!100, inner sep=0pt, minimum width=4pt] (a3) at (0,9) {};
  \node[left] at (0,9) {2};
  \node[circle, draw, fill=black!100, inner sep=0pt, minimum width=4pt] (a4) at (1.5,9) {};
  \node[left] at (1.5,9) {3};
  \node[circle, draw, fill=black!100, inner sep=0pt, minimum width=4pt] (a5) at (-2.5,8) {}; 
  \node[left] at (-2.5,8) {4};
  \node[circle, draw, fill=black!100, inner sep=0pt, minimum width=4pt] (a6) at (-0.5,8) {}; 
  \node[left] at (-0.5,8) {5};
  \node[circle, draw, fill=black!100, inner sep=0pt, minimum width=4pt] (a7) at (-2.5,7) {}; 
  \node[left] at (-2.5,7) {7};
  \node[circle, draw, fill=black!100, inner sep=0pt, minimum width=4pt] (a8) at (1.5,8) {};
  \node[left] at (1.5,8) {6};
  \node[circle, draw, fill=black!100, inner sep=0pt, minimum width=4pt] (a9) at (0.5,7) {};
  \node[left] at (0.5,7) {8};
  \node[circle, draw, fill=black!100, inner sep=0pt, minimum width=4pt] (a10) at (2.5,7) {};
  \node[left] at (2.5,7) {9};
  
  \draw (a1) -- (a2);
  \draw (a1) -- (a3);  
  \draw (a1) -- (a4);  
  \draw (a2) -- (a5);  
  \draw (a2) -- (a6);
  \draw (a5) -- (a7);  
  \draw (a4) -- (a8);  
  \draw (a8) -- (a9);  
  \draw (a8) -- (a10);  
\end{tikzpicture}
\vspace{-7cm}
  \caption{bfsScores for a tree with $n = 9$ edges.}
\end{subfigure}%
\begin{subfigure}{.5\textwidth}
  \centering
    \begin{tikzpicture}[baseline=0, scale=1.1, auto=center]
 \node at (0,10) {$\,\,\,h_n(0) = 1,\, H_n(0) = 1,\, v_n(0) = 10,\, S_n(0) = 0$.};
      \node at (0,9) {$h_n(1) = 3,\, H_n(1) = 4,\, v_n(1) = 9,\, S_n(1) = 6$.};
      \node at (0,8){$\,\,h_n(2) = 3,\, H_n(2) = 7,\, v_n(2) = 6,\, S_n (2)= 15$.};
            \node at (0,7) {$\,\,\,\,h_n(3) = 3,\, H_n(3) = 10,\, v_n(3) = 3,\, S_n(3) = 24$.};
    \end{tikzpicture}
\vspace{-7cm}
  \caption{Values of the random variables defined in \eqref{def:Random variables}.}
\end{subfigure}
\caption{Sample realization of the random variables defined in~\eqref{def:Random variables}.}
\label{fig:bfs_calculation}
\end{figure}
We can express $\totalB(n,\ell )$ in terms of these quantities in the following way:
\begin{align}
\totalB(n,\ell ) = \left|\mathcal{T}_n\right|\cdot \Ebb{S_n(\ell )}
&= C_n \cdot \Ebb{\binom{H_n(\ell )}{2}-\binom{H_n(\ell -1)}{2}} \nonumber\\
&= C_n \cdot \Ebb{\binom{n+1-v_n(\ell +1)}{2}-\binom{n+1- v_n(\ell )}{2}} \nonumber\\
&=C_n\cdot\Ebb{\binom{v_n(\ell +1)}{2} - \binom{v_n(\ell )}{2} - n\lbrb{v_n(\ell +1) - v_n(\ell )}}. \label{eq:totB expectation}      
\end{align}
Moreover, in what follows, we shall use the following results of Tak\'{a}cs \cite[Theorem 4]{Takacs-99}:
\begin{align}
C_n \cdot \Ebb{v_n(\ell )} &= \binom{2n}{n-\ell },
\quad
\text{and} \label{eq: Takacs1}\\
C_n\cdot \Ebb{\binom{v_n(\ell )}{2}} &= (n+\ell )\binom{2n}{n-\ell } - \frac{n+2\ell }{2}
\binom{2n}{n-2\ell }- 2 \ell\sum_{j=\ell }^{2\ell -1}\binom{2n}{n-j}. \label{eq: Takacs2}
\end{align}
 The motivation of Tak\'{a}cs
 was to find explicit expressions for
 the moments of the occupation measure of a 
 Brownian excursion as a limit of the respective
 moments for simple Bernoulli excursions. Using our
notation, the occupation
 measure of the Bernoulli
 excursion is exactly $v_n(\ell )$, and the Brownian
 excursion is the natural encoding of the limit of properly 
 scaled random trees from $\mathcal{T}_n$---that is, the \textit{continuum random tree} introduced by Aldous in
 \cite{CRT1, CRT2, CRT3}. In fact, this link
 allow us to generalize our asymptotics for $\totalB(n,\ell )$
 on $\mathcal{T}_n$ to all classes of trees which
 can be constructed as conditioned Galton--Watson trees (see Section \ref{sec: CRT}). 
 
A formula for $\totalB(n,\ell )$ is obtained below. In the statement of this result, $\Phi$ denotes the cumulative distribution function of the standard normal distribution:
\[
\Phi(s) \coloneqq \frac{1}{\sqrt{2\pi}}\int_{-\infty}^s e^{-x^2/2}\, \D x.
\]
As noted in Section~\ref{sec:summary}, we use the notation $a_n\prec b_n$ to mean that $a_n=o(b_n)$, that is,
$a_n/b_n\to 0$ as $n\to\infty$. Similarly, $a_n\succ b_n$ means that
$a_n=\omega(a_n)$.
\begin{theorem}\label{thm: totalB} 
    For every $n \geq 0$ and
        $0 \leq \ell \leq n$, we have
    \begin{align}
    \label{eq:totalBExpr}
    \totalB(n,\ell ) &= \frac{n(2\ell +1)}{n+\ell +1}\binom{2n}{n-\ell } - \binom{2n}{n-2\ell-1}
    - 2 \sum_{j = \ell + 1}^{2\ell } \binom{2n}{n-j}\\
    &= \totalD(n,\ell ) + \totalD(n,\ell+1 ) - \binom{2n}{n-2\ell -1}
    - 2 \sum_{j = \ell +1}^{2\ell } \binom{2n}{n-j}. \nonumber
    \end{align}
Moreover, if $s=\ell /\sqrt{n}$, we have
   \begin{equation}
       \label{eq:asymptotics totalB}
   \frac{\totalB(n,s\sqrt{n})}{4^n} =
\frac{2s}{\sqrt{\pi}}e^{-s^2}-2\lbrb{\Phi\lbrb{2\sqrt{2}s}-\Phi\lbrb{\sqrt{2}s}}+\bo{\frac{1}{\sqrt{n}}}, 
   \end{equation}
and if $s \succ n^\epsilon$ for some $\epsilon > 0$, we have
\begin{equation}
    \label{eq:exp bound}
 \frac{\totalB(n,s\sqrt{n})}{4^n} 
 \leq 2s\sqrt{n}e^{-n^{2\epsilon}} = \so{e^{-n^\epsilon}}.
\end{equation}

\end{theorem}
\begin{remark}
    The last result shows that the levels with non-negligible $\totalB$ (with respect to the sum of
    all scores) are exactly
    those of order $s\sqrt{n}$ for
    a constant $s$, in the sense that 
    \[
    \text{if $s \prec 1$ or $s \succ 1$, then }
    \frac{\totalB(n,s\sqrt{n})}{4^n} = \so{1}.\]
    More precisely, using \eqref{eq:asymptotics totalB} for $s\prec n^\epsilon$ and
    \eqref{eq:exp bound} for $s\succ n^\epsilon$, we have that
    \begin{equation*}
        \frac{\totalB(n,s\sqrt{n})}{4^n}
        = \begin{cases}
            \bo{\max\{s,n^{-1/2}\}}, &\text{if $s \prec 1$},\\
            \bo{\max\{se^{-s^2},n^{-1/2}\}}, &\text{if $1 \prec s \prec n^\epsilon$},\\
            \so{e^{-n^\epsilon}}, &\text{if $s \succ n^\epsilon$}.
        \end{cases}
    \end{equation*}
\end{remark}


\begin{proof}[Proof of Theorem \ref{thm: totalB}]
    Substituting \eqref{eq: Takacs1} and \eqref{eq: Takacs2} into \eqref{eq:totB expectation}, we get
\begin{small}\begin{align}\label{eq:totalB calc}
        \totalB(n,\ell ) &= (n+\ell +1)\binom{2n}{n-\ell -1} - (n+\ell )\binom{2n}{n-\ell }
        -\frac{n+2\ell +2}{2}\binom{2n}{n-2\ell -2}+\frac{n+2\ell }{2}\binom{2n}{n-2\ell } \nonumber\\
        &\hphantom{{}=
(n}
-2(\ell +1)\sum_{j=\ell +1}^{2\ell +1} \binom{2n}{n-j}+ 2\ell \sum_{j=\ell }^{2\ell -1}
        \binom{2n}{n-j}- n\lbrb{\binom{2n}{n-\ell -1}
        - \binom{2n}{n-\ell }} \nonumber\\
        &= \lbrb{\ell +1}\binom{2n}{n-\ell -1} + l\binom{2n}{n-\ell }
        \nonumber\\
        &\hphantom{{}=
(n}+ \binom{2n}{n-2\ell -1}\lbrb{-\frac{n-2\ell -1}{2}
        + \frac{n+2\ell +1}{2}
        -2\ell -2} 
        -2 \sum_{j=\ell +1}^{2\ell}\binom{2n}{n-j} \nonumber \\
        &= \frac{n(2\ell +1)}{n+\ell +1}\binom{2n}{n-\ell }
        - \binom{2n}{n-2\ell -1} -2 \sum_{j=\ell +1}^{2\ell }\binom{2n}{n-j}
\end{align}\end{small}as needed. The fact that the first term in the last expression 
    is equal to $\totalD(n,\ell ) + \totalD(n,\ell +1)$ can be seen from the
    second equality above and Theorem \ref{th:totalD}.

    For the remaining asymptotic statements, we will use two additional facts giving us asymptotic bounds for $\binom{2n}{n-\ell }$ and $\sum_{j=\ell +1}^{2\ell }\binom{2n}{n-j}$, respectively. The first fact is a variant of the local limit theorem \cite[Chapter VII, Theorem 13]{Petrov-75}, which states that for 
     a binomial random variable
    $\mathcal{B}_{2n} \sim \operatorname{Bin}(2n, 1/2)$, we have
    \[
    \sup_{\ell }\,\left|\P\lbrb{\mathcal{B}_{2n} = n-\ell } - \frac{1}{\sqrt{\pi n}}e^{-\ell^2/n} \right|\leq \frac{c_0}{n^{3/2}}
    \]
    for some constant $c_0 > 0$;
    this implies
    \begin{equation}
        \label{eq:local limit}
    \binom{2n}{n-\ell } = 4^n \lbrb{\frac{1}{\sqrt{\pi n}}e^{-\ell^2/n} + \bo{\frac{1}{n^{3/2}}}} = 4^n\lbrb{\bo {\frac{1}{\sqrt{n}}}}.
    \end{equation}
     The second fact is the standard central limit theorem, which along
    with the Berry--Essen bound (see \cite[p.62]{Shiryaev-16}), implies that
    \begin{equation}
        \label{eq:Berry-Essen}
   \sup_{x}\,\left| 
    \P\lbrb{\frac{\mathcal{B}_{2n}-n}{\sqrt{n/2}} \leq x} - \Phi\lbrb{x}\right| = 
    \sup_{x}\, \left| 
    \P\lbrb{\mathcal{B}_{2n} \leq x} - \Phi\lbrb{\frac{x-n}{\sqrt{n/2}}}\right| \leq \frac{1}{\sqrt{2n}}.
    \end{equation}
    Therefore, using $s = \ell/\sqrt{n}$, we have
     \begin{align}
     \mkern-48mu
\sum_{j=\ell +1}^{2\ell }\binom{2n}{n-j} =
4^n \P\lbrb{n+\ell +1 \leq \mathcal{B}_{2n} \leq n+2\ell }
    = 4^n\lbrb{\Phi\lbrb{2\sqrt{2}s}-\Phi\lbrb{\sqrt{2}s}+\bo{\frac{1}{\sqrt{n}}}}. \label{e-2nnj}
    \end{align}
    Substituting \eqref{eq:local limit} and \eqref{e-2nnj} into \eqref{eq:totalB calc}, we obtain
    \begin{equation}\label{eq:expand totalB}
        \begin{split}
              \frac{1}{4^n}\totalB(n,s\sqrt{n}) &= 
    \frac{2s\sqrt{n}+1}{1+(s\sqrt{n}+1)/n}\lbrb{\frac{1}{\sqrt{\pi n}}e^{-s^2} + \bo{\frac{1}{n^{3/2}}}} \\
    &
    \hphantom{{}=
\frac{2s\sqrt{n}+1}{1+(s\sqrt{n}+1)/n}\frac{2s\sqrt{n}+1}{1}}
-2\lbrb{\Phi\lbrb{2\sqrt{2}s}-\Phi\lbrb{\sqrt{2}s}}+\bo{\frac{1}{\sqrt{n}}}.
        \end{split}
    \end{equation}
In fact, \eqref{eq:expand totalB} holds for every $s$ that is a function of $n$. To further expand, we first note that if $\ell \succ n^{1/2+\epsilon}$ for
some $\epsilon>0$, or equivalently if $s \succ n^\epsilon$, then $\totalB(n,\ell )$ is very small. Indeed, we can get the following inequalities:
    \begin{equation}
        \label{eq:totalB in proof}
    \frac{1}{4^n}\totalB(n,\ell ) \leq
   \frac{n(2\ell +1)}{n+\ell +1} \frac{1}{4^n}\binom{2n}{n-\ell } \leq
n\P\lbrb{\mathcal{B}_{2n} \geq n +\ell } \leq
ne^{-\ell^2/n} \leq n e^{-n^{2\epsilon}} = \so{e^{-n^\epsilon}}.
    \end{equation}
    For the first inequality, we upper-bounded $\totalB(n,\ell )$ by the first term of the expression obtained in~\eqref{eq:totalB calc}, since it is the only positive term. For the second inequality, we use 
    $$\frac{1}{4^{n}}\binom{2n}{n-\ell } = \P\lbrb{\mathcal{B}_{2n} = n +\ell }\leq \P\lbrb{\mathcal{B}_{2n} \geq n +\ell }.$$ 
    Finally, the third inequality follows from Hoeffding's inequality \cite[Theorem 1]{Hoeffding-63}, and thus we have established \eqref{eq:exp bound}.
    Further, using the Taylor expansion
    $1/(1+x) = 1 + \bo{x}$ in \eqref{eq:expand totalB}, we obtain
    \begin{equation*}
        \begin{split}
             \frac{1}{4^n}\totalB(n,s\sqrt{n}) &= 
\lbrb{2s\sqrt{n} + 1}\lbrb{1+\bo{\frac{s}{\sqrt{n}}} + \bo{\frac1n}}\lbrb{\frac{1}{\sqrt{\pi n}}e^{-s^2} + \bo{\frac{1}{n^{3/2}}}} \\
 &
    \hphantom{{}=
\frac{2s\sqrt{n}+1}{1+(s\sqrt{n}+1)/n}\frac{2s\sqrt{n}+1}{1}}
-2\lbrb{\Phi\lbrb{2\sqrt{2}s}-\Phi\lbrb{\sqrt{2}s}}+\bo{\frac{1}{\sqrt{n}}}\\
&= \frac{2se^{-s^2}}{\sqrt{\pi}}-2\lbrb{\Phi\lbrb{2\sqrt{2}s}-\Phi\lbrb{\sqrt{2}s}} +\bo{\frac{1}{\sqrt{n}}},
        \end{split}
    \end{equation*}
    where after expanding the former
    equation, we use
     the inequality $s \leq \sqrt{n}$ and the fact that $se^{-s^2} $ and
     $s^2e^{-s^2} $ are bounded functions of $s$.
\end{proof}

\section{Comparing the average time complexity of BFS and DFS.}
\label{sec:threshold}
Now, we are ready to prove the main result of this paper.
\begin{proof}[Proof of Theorem~\ref{th:threshold}]
    Combining the results of
    Theorems \ref{th:totalD} and
    \ref{thm: totalB}, we have
\[
    f_n(\ell ) \coloneqq \totalB(n,\ell ) - \totalD(n,\ell ) = (\ell +1)\binom{2n}{n-\ell -1} - \binom{2n}{n-2\ell -1} - 2 \sum_{j = \ell +1}^{2\ell } \binom{2n}{n-j}.
\]
Therefore, by Equation~\eqref{eq:asymptotics totalB} and the Stirling approximation for $\totalD(n,\ell )$, we have
\[f_n\lbrb{s\sqrt{n}}\sim \frac{4^n}{\sqrt{\pi}} \left( s e^{-s^2} - 2\sqrt{\pi}\left( \Phi\left(2\sqrt{2}s\right) - \Phi\lbrb{\sqrt{2}s} \right) \right)
\]
when $s\geq 0$ and $n\to \infty$. It is readily verified that the function in the brackets has a single positive root, which establishes the second part of the theorem provided that the threshold is unique.

We note that it can be checked that if $f_n(\ell ) = 0$ with some $\ell \in (0, n)$,
then if $n$ is large enough, $\ell$ should be of order $\sqrt{n}$, and therefore
we are in the scenario considered above. However, as we have observed that the
threshold is unique for small values of $n$ too, we prove that it exists for all
$n$.

Let $n\geq 3$. First note that $f_n(1)< 0=f_n(0)$ and $f_n(n-1)>0=f_n(n)$.
We will show that $g_n(\ell ) \coloneqq f_n(\ell ) - f_n(\ell -1)$
changes sign two times, and because $f_n$ is
 decreasing around $0$ and $n$, this would mean
 that $f_n(\ell -1)\leq 0$ and $f_n(\ell ) > 0$ for a 
 unique $\ell\in(0,n)$. A brief calculation
 shows that
\[g_n(\ell ) =f_n(\ell ) - f_n(\ell -1) = \frac{-2\ell^2+3n+2}{n+\ell +1}\binom{2n}{n+\ell } - \binom{2n+2}{n+2\ell +1}.
\] 
Denote the first term $A(\ell )$ and the second $B(\ell )$, 
where we skip the dependence on $n$ for
simpler notation. Since $g_n$ is negative for $2\ell^2\geq 3n+2$, let us consider $\ell$ for which $A(\ell )>0$. Moreover, consider only the non-trivial cases where $B(\ell )>0$, i.e., $2\ell < n$. This will be implied by the previous restriction when $n \geq 7$. In these cases, observe that
\[
\sign\lbrb{g_n(\ell )} = \sign\lbrb{A(\ell )-B(\ell )} = \sign\lbrb{\frac{A(\ell )}{B(\ell )} - 1}.
\]
Therefore, setting $D(\ell ) \coloneqq A(\ell )/B(\ell )$, we have to
solve $D(\ell ) > 1$.
We have that
\[
\frac{D(\ell +1)}{D(\ell )}  = \frac{(-2(\ell +1)^2 + 3n + 2)(n+2\ell +2)(n+2\ell +3)(n-\ell )}{(-2\ell^2 + 3n + 2)(n+\ell +2)(n-2\ell )(n-2\ell +1)},
\]
so $D(\ell +1)>D(\ell )$ if and only if
\begin{multline*}
\quad h(\ell )\coloneqq  n^3(14\ell +4) + n^2(-12\ell^3-16\ell^2+28\ell +6) \\
+n(-52\ell^3-76\ell^2-12\ell -4) +16\ell^5+48\ell^4+36\ell^3+12\ell^2
+8\ell > 0. \quad
\end{multline*}
Differentiating, we obtain
\begin{align*}
h'(\ell) &= 14n^3 + n^2(-36\ell^2\!-\!32\ell \!+\!28)+n(-156\ell^2\!-\!152\ell\!-\!12)+80\ell^4\!+\!192\ell^3\!+\!108\ell^2\!+\!24\ell\!+\!8,\\
h''(\ell ) &= n^2(-72\ell \!-\!32)+n(-312\ell \!-\!152)+320\ell^3\!+\!576\ell^2\!+\!216\ell \!+\!24\\
&= \ell (-72n^2\!-\!312n\!+\!320\ell^2\!+\!216)-32n^2\!-\!152n\!+\!576\ell^2\!+\!24\\
&<\ell \lbrb{-72n^2\!+\!168 n\!+\!536} -32n^2 \!+\! 712n \!+\! 600 < 0
\end{align*}
for $n \geq 10$ and where we used $2\ell^2 < 3n+2$. (The statement of the theorem, that is uniqueness of the threshold, for $n < 10$ can be verified directly.)
Therefore, $h$ is strictly concave. A direct substitution shows
that $h(0) > 0$ and 
$h(6\sqrt{n}/5) <0$, so
$D(\ell )$ initially increases,
then decreases. As
$D(1) <1 $, this means that $g_n$ changes sign $0$ or $2$ times, but the former is impossible since
$\sum_k g_n(k) = 0$ and $g_n(1)<0$. Thus our proof is complete.
\end{proof}
Theorem~\ref{th:threshold} gives us a unique threshold occurring at $\lambda\sqrt{n}$, where $\lambda\approx 0.789004$. It is natural to suspect that this threshold will be close to the average level of a node in $\mathcal{T}_n$, when $n\to\infty$. In order to make this comparison, we will use the following result of Flajolet and Sedgewick.
\begin{theorem}[{\cite[Prop. VII.3]{flajSed}}]
    The average level of the nodes among all trees in $\mathcal{T}_{n}$ is $\frac{1}{2}\sqrt{\pi n} + \bo{1}$.
\end{theorem}
 We have that $\frac{1}{2}\sqrt{\pi}\approx 0.8862>0.789\approx\lambda$ and thus for ordered trees, the unique threshold occurs around the average level for the nodes in $\mathcal{T}_{n}$, but slightly below it. In Section~\ref{sec:q}, we ask whether this is a general phenomenon that applies to other families of trees. For instance, is there a similar transition for \emph{binary trees} around the average level for the nodes in them? We also ask for an explanation of the obtained small discrepancy in the case of ordered trees.
\subsection{The BFS and DFS complexity ratio for certain intervals of target levels}
\label{subsec:ratios}
We will also show that the ratios between the average complexities of BFS and DFS are the constants $0$ and $2$, in two large intervals for $\ell$ around $\ell = \sqrt{n}$. 

First, observe that Theorem~\ref{th:totalD} and the local limit theorem (see Equation~\eqref{eq:local limit}) imply that for any $\ell = s\sqrt{n}$,
\begin{equation}
\frac{\totalD(n,s\sqrt n)}{4^n}\longrightarrow g(s):=\frac{s}{\sqrt\pi}e^{-s^2}.
\end{equation}
In addition, recall that Equation~\eqref{eq:asymptotics totalB} gives that for any $\ell = s\sqrt{n}$,
\begin{equation}
\frac{\totalB(n,s\sqrt n)}{4^n} \;\longrightarrow\; B(s) := \frac{2s}{\sqrt\pi}e^{-s^2} - 2\Delta(s) = 2(g(s)-\Delta(s)), \qquad n\to\infty,
\label{eq:Bform}
\end{equation}
where
\begin{equation}
\Delta(s) := \Phi(2\sqrt2\,s) - \Phi(\sqrt2\,s).
\end{equation}

We then have 
\begin{equation}
\label{eq:BDratio}
\frac{\totalB(n,s\sqrt n)}{\totalD(n,s\sqrt n)}\to \frac{B(s)}{g(s)} = 2\left(1-\frac{\Delta(s)}{g(s)}\right).    
\end{equation}
To investigate this ratio at $s\to 0$, we will need the expansion of $\Phi$ near $0$ to third order.

\begin{lemma}[Taylor expansion of $\Phi$ at $0$]\label{lem:taylor}
As $x\to0$,
\[
\Phi(x) = \frac12 + \frac{1}{\sqrt{2\pi}}\left(x-\frac{x^3}{6}\right) + O(x^5).
\]
\end{lemma}
\begin{proof}
By definition $\Phi(x)=\frac12+\frac{1}{\sqrt{2\pi}}\int_0^x e^{-t^2/2}\,dt$. The MacLaurin series of the integrand is,
\[
e^{-t^2/2} = 1-\frac{t^2}{2}+\frac{t^4}{8}-O(t^6).
\]
Integrating term by term over $[0,x]$ gives us
\[
\int_0^x e^{-t^2/2}\,dt = x - \frac{x^3}{6} + \frac{x^5}{40} - O(x^7).
\]
Substituting in the definition of $\Phi$ gives the claim.
\end{proof}

Apply Lemma~\ref{lem:taylor} at $x=2\sqrt2\,s$ and $x=\sqrt2\,s$. Using $(2\sqrt2\,s)^3 = 8\cdot2\sqrt2\,s^3=16\sqrt2\,s^3$ and $(\sqrt2\,s)^3=2\sqrt2\,s^3$,
\[
\Phi(2\sqrt2\,s) = \frac12+\frac{1}{\sqrt{2\pi}}\left(2\sqrt2\,s-\frac{16\sqrt2}{6}s^3\right)+O(s^5)
= \frac12+\frac{1}{\sqrt{2\pi}}\left(2\sqrt2\,s-\frac{8\sqrt2}{3}s^3\right)+O(s^5),
\]
\[
\Phi(\sqrt2\,s) = \frac12+\frac{1}{\sqrt{2\pi}}\left(\sqrt2\,s-\frac{2\sqrt2}{6}s^3\right)+O(s^5)
= \frac12+\frac{1}{\sqrt{2\pi}}\left(\sqrt2\,s-\frac{\sqrt2}{3}s^3\right)+O(s^5).
\]
After simplification, we obtain that when $s\to0$,
\begin{equation*}
\Delta(s) = \Phi(2\sqrt2\,s)-\Phi(\sqrt2\,s) = \frac{1}{\sqrt{2\pi}}\left(\sqrt2\,s-\frac{7\sqrt2}{3}s^3\right)+O(s^5) = \frac{s}{\sqrt\pi} - \frac{7}{3\sqrt\pi}\,s^3 + O(s^5).
\end{equation*}

Meanwhile 
$$g(s)=\dfrac{s}{\sqrt\pi}e^{-s^2}=\dfrac{s}{\sqrt\pi}\bigl(1-s^2+O(s^4)\bigr)=\dfrac{s}{\sqrt\pi}-\dfrac{s^3}{\sqrt\pi}+O(s^5).$$ 
Hence in $B(s)=2g(s)-2\Delta(s)$ the linear terms cancel exactly, and only the cubic terms survive:
\begin{equation*}
B(s) = 2\left(-\frac{s^3}{\sqrt\pi}\right)-2\left(-\frac{7}{3\sqrt\pi}s^3\right)+O(s^5) = \left(-2+\frac{14}{3}\right)\frac{s^3}{\sqrt\pi}+O(s^5) = \frac{8}{3\sqrt\pi}\,s^3+O(s^5).
\end{equation*}
Consequently, when $s\to 0$,
\begin{equation}
\frac{\totalB(n,s\sqrt n)}{\totalD(n,s\sqrt n)} \ \longrightarrow\ \frac{B(s)}{g(s)} \sim \frac{8}{3}s^2 \ \longrightarrow\ 0.
\label{eq:ratio0}
\end{equation}
In fact, this conclusion is not limited to a fixed constant $s$, but extends to any sequence of target levels tending to $0$ relative to $\sqrt n$, at a suitable rate, as follows.
\begin{theorem}
\label{th:43}
Let $\ell=\ell(n)$ be a sequence of target levels satisfying $n^{1/3}\prec\ell(n)\prec\sqrt n$ (equivalently, $s(n):=\ell(n)/\sqrt n$ satisfies $n^{1/6}s(n)\to\infty$ and $s(n)\to0$). Then
\[
\frac{\mathrm{totalB}(n,\ell(n))}{\mathrm{totalD}(n,\ell(n))} \longrightarrow 0, \qquad n\to\infty.
\]
\end{theorem}
\begin{proof}
Equation~\eqref{eq:asymptotics totalB} holds for every $s$ that is a function of $n$ — its two ingredients, the local limit theorem and the Berry\textendash Esseen bound, are $\sup_\ell$- and $\sup_x$-bounds with a universal constant, not statements calibrated to one fixed $s$. Hence, uniformly in $\ell=o(n)$,
\[
\frac{\mathrm{totalB}(n,\ell)}{4^n} = B(s)+O(n^{-1/2}), \qquad \frac{\mathrm{totalD}(n,\ell)}{4^n} = g(s)+O(n^{-1/2}), \qquad s=\ell/\sqrt n,
\]
with implied constants independent of $s$. Since $s(n)\to0$, we have $B(s)\sim\frac{8}{3\sqrt\pi}s^3$ near $0$, so the $O(n^{-1/2})$ error term above is dominated by $B(s(n))$ precisely when $s(n)^3\succ n^{-1/2}$, i.e.\ when $n^{1/6}s(n)\to\infty$, equivalently $\ell(n)\succ n^{1/3}$; and since $g\succ B$ near $0$, the same rate suffices for the error to be negligible relative to $g(s(n))$ as well. Under the hypothesis $n^{1/3}\prec\ell(n)\prec\sqrt n$, both error terms become relative $o(1)$ corrections, and
\[
\frac{\mathrm{totalB}(n,\ell(n))}{\mathrm{totalD}(n,\ell(n))} = \frac{B(s(n))\bigl(1+o(1)\bigr)}{g(s(n))\left(1+o(1)\right)} \sim \frac83 s(n)^2 \longrightarrow 0. \qedhere
\]
\end{proof}

When $s\to\infty$, we will need the Lemma below which can be proven via integration by parts. 
\begin{lemma}[Mills' ratio, \cite{Mills1926}]
\label{lem:mills}
For $\varphi(x)=\frac1{\sqrt{2\pi}}e^{-x^2/2}$, as $x\to\infty$,
\[
1-\Phi(x) = \frac{\varphi(x)}{x}\bigl(1+O(x^{-2})\bigr).
\]
\end{lemma}
\noindent Write $\Delta(s)$ as a difference of two upper tails,
\[
\Delta(s) = \Phi(2\sqrt2\,s)-\Phi(\sqrt2\,s) = \bigl[1-\Phi(\sqrt2\,s)\bigr] - \bigl[1-\Phi(2\sqrt2\,s)\bigr],
\]
and apply Lemma~\ref{lem:mills}. Since $(\sqrt2\,s)^2/2 = s^2$, we have $\varphi(\sqrt2\,s)=\frac{1}{\sqrt{2\pi}}e^{-s^2}$ and  thus
\[
1-\Phi(\sqrt2\,s) = \frac{\varphi(\sqrt2\,s)}{\sqrt2\,s}\bigl(1+O(s^{-2})\bigr)
= \frac{e^{-s^2}}{\sqrt{2\pi}\cdot\sqrt2\,s}\bigl(1+O(s^{-2})\bigr)
= \frac{e^{-s^2}}{2\sqrt\pi\,s}\bigl(1+O(s^{-2})\bigr).
\]
For the second tail, since $(2\sqrt2\,s)^2/2=4s^2$, we have $\varphi(2\sqrt2\,s)=\frac1{\sqrt{2\pi}}e^{-4s^2}$, so
\[
1-\Phi(2\sqrt2\,s) = \frac{\varphi(2\sqrt2\,s)}{2\sqrt2\,s}\bigl(1+O(s^{-2})\bigr)
= \frac{e^{-4s^2}}{\sqrt{2\pi}\cdot2\sqrt2\,s}\bigl(1+O(s^{-2})\bigr)
= \frac{e^{-4s^2}}{4\sqrt\pi\,s}\bigl(1+O(s^{-2})\bigr).
\]
The second tail is exponentially smaller than the first since when $s\to\infty$ their ratio is
\[
\frac{e^{-4s^2}/(4\sqrt\pi\,s)}{e^{-s^2}/(2\sqrt\pi\,s)} = \frac{1}{2}\,e^{-3s^2} \longrightarrow 0
\]
Thus, after subtracting we obtain
\begin{equation}
\Delta(s) = \frac{e^{-s^2}}{2\sqrt\pi\,s}\bigl(1+O(s^{-2})\bigr) - \frac{e^{-4s^2}}{4\sqrt\pi\,s}\bigl(1+O(s^{-2})\bigr) = \frac{e^{-s^2}}{2\sqrt\pi\,s}\bigl(1+O(s^{-2})\bigr).
\end{equation}
\noindent Therefore,
\[
\frac{\Delta(s)}{g(s)} = \frac{\dfrac{e^{-s^2}}{2\sqrt\pi\,s}\bigl(1+O(s^{-2})\bigr)}{\dfrac{s}{\sqrt\pi}e^{-s^2}}
= \frac{1}{2\sqrt\pi\,s}\cdot\frac{\sqrt\pi}{s}\,\bigl(1+O(s^{-2})\bigr) = \frac{1}{2s^2}\bigl(1+O(s^{-2})\bigr) \ \longrightarrow\ 0.
\]
In particular $\Delta(s)=o(g(s))$. Hence, when $s\to\infty$, by Equation~\eqref{eq:BDratio} we have
\begin{equation}
B(s) \sim 2g(s)\qquad\text{ and } \qquad \frac{\totalB(n,s\sqrt n)}{\totalD(n,s\sqrt n)} \ \longrightarrow\ 2.
\label{eq:ratioinf}
\end{equation}
As with the previous limit, this conclusion is not limited to a fixed $s$, and extends as follows.
\begin{theorem}
\label{th:45}
Let $\ell=\ell(n)$ be a sequence of target levels satisfying $\sqrt n\prec\ell(n)\prec\sqrt{n\ln n}$ (equivalently, $s(n):=\ell(n)/\sqrt n$ satisfies $s(n)\to\infty$ and $s(n)/\sqrt{\ln n}\to0$). Then
\[
\frac{\mathrm{totalB}(n,\ell(n))}{\mathrm{totalD}(n,\ell(n))} \longrightarrow 2, \qquad n\to\infty.
\]
\end{theorem}
\begin{proof}
By the same uniform form of Equation~\eqref{eq:asymptotics totalB} used in the proof of Theorem~\ref{th:43}, it suffices to check that $g(s(n))$ dominates the $O(n^{-1/2})$ error term, since $B(s)\sim2g(s)$ and $\Delta(s)=o(g(s))$ as $s\to\infty$ by (22). Now
\[
g(s) = \frac{s}{\sqrt\pi}e^{-s^2} \succ n^{-1/2} \iff s^2 \prec \tfrac12\ln n + \ln s + O(1),
\]
and since $\ln s(n) = o(\ln n)$, this holds whenever $s(n)^2=o(\ln n)$, i.e.\ $s(n)=o\bigl(\sqrt{\ln n}\bigr)$, equivalently $\ell(n)=o\bigl(\sqrt{n\ln n}\bigr)$. Under this condition both $\mathrm{totalB}(n,\ell)/4^n$ and $\mathrm{totalD}(n,\ell)/4^n$ reduce to their leading terms $B(s(n))\sim2g(s(n))$ and $g(s(n))$, giving
\[
\frac{\mathrm{totalB}(n,\ell(n))}{\mathrm{totalD}(n,\ell(n))} \longrightarrow 2. \qedhere
\]
\end{proof}
 
Together, Theorems~\ref{th:43} and~\ref{th:45} show that the ratio $\mathrm{totalB}(n,\ell)/\mathrm{totalD}(n,\ell)$ already exhibits its limiting behavior $0$, respectively $2$, for every target level in certain large intervals. The exponents $\tfrac13$ and $\tfrac12$ in these two theorems are not claimed to be sharp. In fact, we suspect that the two theorems hold whenever $\ell\prec \sqrt{n}$ and  $\ell\succ \sqrt{n}$, respectively. Yet, we were not able to show that. One reason for us to believe the mentioned extension is possible (at least for the latter interval) is Theorem~\ref{th:dTruncBFSratio}, its possible extension mentioned in Remark~\ref{rem:main}, and the fact that the truncated DFS algorithm is almost identical to the DFS algorithm for large target levels $\ell$.
\section{The generating function for $\totalB$ and Fibonacci polynomials}
\label{sec:fibo}
In this section, we obtain an expression for the generating function of $\totalB(n,\ell )$, without using the results in the previous Section~\ref{sec:totalB}. Instead, we use a generating function approach which demonstrates a surprising connection between the generating function of $\totalB(n,\ell )$ and Fibonacci polynomials. Furthermore, we leverage this connection to find a formula for $\totalB(n,\ell )$ as a difference of two alternating sums, which is different from the one from Equation~\eqref{eq:totalBExpr}. We begin by defining the generating functions
\begin{equation*}
F_{\ell }(x,y,z) \coloneqq \sum_{n=0}^{\infty} \sum\limits_{T\in \mathcal{T}_n} x^{k(T)}y^{m(T)}z^{n},
\end{equation*}
where $k(T)$ is the number of non-root nodes in $T$ at levels smaller than $\ell$ and $m(T)$ is the number of nodes in $T$ at level $\ell$. Now, consider the ordinary generating function
\begin{equation*}
C=C(z)\coloneqq \sum_{n=0}^{\infty} C_n z^n = \frac{1-\sqrt{1-4z}}{2z} = \frac{2}{1+\sqrt{1-4z}}
\end{equation*}
 for the Catalan numbers $C_n$. Since $C_n$ counts ordered trees with $n$ edges and every ordered tree has a single node (the root) at level 0, it follows that 
\begin{equation*}
F_0 = yC = \frac{2y}{1+\sqrt{1-4z}}.
\end{equation*}
Using one of the standard decompositions of ordered trees, which gives that the children of the root of every ordered tree are themselves roots of ordered trees, implies the equations
\begin{equation}
\label{eq:F-1}
F_{1} = 1 + zF_{0} + (zF_{0})^2 + \cdots = \frac{1}{1-zF_{0}} = \frac{1}{1-yzC}
\end{equation}
and
\begin{equation}
\label{eq:F-l}
F_{\ell } = 1 + xzF_{\ell -1} + (xzF_{\ell -1})^2 + \cdots = \frac{1}{1-xzF_{\ell -1}} \quad (\ell\geq 2).
\end{equation}
This recursive relation shows that $F_\ell$ can be written as a continued fraction as, for example, in the work of de Bruijn, Knuth and Rice \cite{de1972average}. The formula for $F_\ell$ that one can obtain via this approach, will be proven directly in Section~\ref{subsec:genfuncs}. Then, since our goal is to find $\totalB(n,\ell )$, we will use the obtained formula for $F_\ell$ to find the generating function
\begin{equation*}
B_{\ell }=B_{\ell }(z)\coloneqq\sum_{n=0}^{\infty}\totalB(n,\ell )z^{n}.
\end{equation*}
If $m(T)=m$ and $k(T)=k$, then a straightforward counting argument yields 
\begin{equation*}
\sum\limits_{v\in \lev(T,\ell )} \bfsScore(v) = mk + \frac{m(m+1)}{2}.
\end{equation*}
The latter equation gives us the following expression for $B_\ell$ in terms of $F_\ell$ and its derivatives.
\begin{equation} \label{e-BF}
B_{\ell }=\left[\Big(\frac{\partial^{2}}{\partial x\partial y}+\frac{1}{2}\cdot\frac{\partial^{2}}{\partial y^{2}}+\frac{\partial}{\partial y}\Big)\,F_{\ell }\right]_{\substack{x=1,\,y=1}
}.
\end{equation}

\subsection{Fibonacci polynomials and some identities}
\label{subsec:lemmas}
In this section, we prove some lemmas that will be needed to show that the generating functions $F_{\ell }$ and $B_{\ell }$ can be expressed in terms of the \textit{Fibonacci polynomials} $f_{n}(z)$ defined by
\[
f_{n}(z)\coloneqq\sum_{k=0}^{\left\lfloor n/2\right\rfloor }{n-k \choose k}z^{k}
\] 
for $n\geq0$ (and $f_{n}(z)\coloneqq 0$ for $n<0$). The first few polynomials $f_{n}(z)$ are displayed in Table~\ref{tb-fibpoly}. It is well known that the $n$th Fibonacci number $f_n$ (defined with initial terms $f_0=1$ and $f_1=1$) counts square-domino tilings of a $1\times n$ rectangle, and the coefficient ${n-k \choose k}$ of $f_{n}(z)$ counts such tilings with $k$ dominoes and $n-2k$ squares.\footnote{We note that \cite{benj} and some other sources define Fibonacci polynomials slightly differently, by taking the coefficient of $z^n$ in $f_{n}(z)$ to be the number of square-domino tilings of a $1\times n$ rectangle with $k$ squares and $(n-k)/2$ dominoes.} Thus setting $z=1$ in $f_{n}(z)$ recovers the $n$th Fibonacci number $f_n$.

\begin{table}[!h]
\begin{center}
\renewcommand{\arraystretch}{1.3}%
\begin{tabular}{|c|c|c|c|c|}
\cline{1-2} \cline{2-2} \cline{4-5} \cline{5-5} 
$n$ & $f_{n}(z)$ &  & $n$ & $f_{n}(z)$\tabularnewline
\cline{1-2} \cline{2-2} \cline{4-5} \cline{5-5} 
$0$ & $1$ &  & $5$ & $1+4z+3z^{2}$\tabularnewline
\cline{1-2} \cline{2-2} \cline{4-5} \cline{5-5} 
$1$ & $1$ &  & $6$ & $1+5z+6z^{2}+z^{3}$\tabularnewline
\cline{1-2} \cline{2-2} \cline{4-5} \cline{5-5} 
$2$ & $1+z$ &  & $7$ & $1+6z+10z^{2}+4z^{3}$\tabularnewline
\cline{1-2} \cline{2-2} \cline{4-5} \cline{5-5} 
$3$ & $1+2z$ &  & $8$ & $1+7z+15z^{2}+10z^{3}+z^{4}$\tabularnewline
\cline{1-2} \cline{2-2} \cline{4-5} \cline{5-5} 
$4$ & $1+3z+z^{2}$ &  & $9$ & $1+8z+21z^{2}+20z^{3}+5z^{4}$\tabularnewline
\cline{1-2} \cline{2-2} \cline{4-5} \cline{5-5} 
\end{tabular}
\par\end{center}
\caption{The first few Fibonacci polynomials.} \label{tb-fibpoly}
\end{table}

Next, we prove several technical lemmas related to Fibonacci polynomials.

\begin{lemma}
\label{l-frec}
For all $n\geq0$, we have
$f_{n+1}(z)=f_{n}(z)+zf_{n-1}(z)$.
\end{lemma}

\begin{proof}
This follows from the combinatorial proof that the number of square-domino tilings of a $1\times n$ rectangle satisfy the Fibonacci recurrence; the term $f_{n-1}(z)$ is multiplied by the variable $z$ (keeping track of dominoes), because that corresponds to the case where we add a domino to the end of a $1\times(n-1)$ tiling to get a $1\times(n+1)$ tiling.
\end{proof}

\begin{lemma}
\label{l-lclike}
For all $n\geq0$, we have
$f_{n}(z)^{2}-f_{n-1}(z)f_{n+1}(z)=(-z)^{n}$.
\end{lemma}

\begin{proof}
This formula is a straightforward refinement of \textit{Cassini's formula} $f_{n}^{2}-f_{n-1}f_{n+1}=(-1)^{n}$ for the Fibonacci numbers $f_{n}$, and the ``tail-swapping'' combinatorial proof of Cassini's formula (see \cite[Identity 8]{benj}) suffices to prove this refinement as well.
\end{proof}

\begin{lemma}
\label{l-denomC}
For all $n\geq0$, we have
$\left(f_{n}(-z)-f_{n-1}(-z)zC\right)^{-1}=C^{n}$.
\end{lemma}

\begin{proof}
The case $n=0$ reduces to the statement $(1-0\cdot zC)^{-1}=1$, which is obviously true, and the case $n=1$ reduces to 
$(1-zC)^{-1}=C$, which is equivalent to the functional equation $C=1+zC^{2}$ for the Catalan generating function. Now, assume that 
$\left(f_{n}(-z)-f_{n-1}(-z)zC\right)^{-1}=C^{n}$
for a fixed $n\geq1$. We then have
\begin{align*}
C^{n+1} & =\left(f_{n}(-z)-f_{n-1}(-z)zC\right)^{-1}(1-zC)^{-1}\\
 & =\left(f_{n}(-z)-f_{n-1}(-z)zC-f_{n}(-z)zC+f_{n-1}(-z)z^{2}C^{2}\right)^{-1}\\
 & =\left(f_{n}(-z)-zf_{n-1}(-z)+zf_{n-1}(-z)-f_{n-1}(-z)zC-f_{n}(-z)zC+f_{n-1}(-z)z^{2}C^{2}\right)^{-1}\\
 & =\left(f_{n+1}(-z)-f_{n}(-z)zC+zf_{n-1}(-z)\left(1-C+zC^{2}\right)\right)^{-1}\\
 & =\left(f_{n+1}(-z)-f_{n}(-z)zC\right)^{-1},
\end{align*}
and the result follows by induction.
\end{proof}

\begin{lemma}
\label{l-fderiv}
For all $n\geq0$, we have
\[
nf_{n}(-z)-2z\frac{\D}{\D z}f_{n}(-z)=-\frac{\D}{\D z}f_{n+1}(-z).
\]
\end{lemma}

\begin{proof}
Consider the ordinary generating function
\[
G(z,x)\coloneqq\sum_{n=0}^{\infty}f_{n}(z)x^{n}=\frac{1}{1-x-zx^{2}}
\]
for the Fibonacci polynomials. Taking the partial derivatives of $G(-z,x)$ yields
\[
\frac{\partial}{\partial x}G(-z,x)=\frac{1-2zx}{(1-x+zx^{2})^2}\quad\text{and}\quad\frac{\partial}{\partial z}G(-z,x)=\frac{-x^{2}}{(1-x+zx^{2})^2},
\]
so we have the partial differential equation
\[
x^{2}\frac{\partial}{\partial x}G(-z,x)=(2zx-1)\frac{\partial}{\partial z}G(-z,x),
\]
which is equivalent to
\begin{equation}
x\frac{\partial}{\partial x}G(-z,x)-2z\frac{\partial}{\partial z}G(-z,x)=-\frac{1}{x}\frac{\partial}{\partial z}G(-z,x).\label{e-diffeq}
\end{equation}
Extracting coefficients of $x^{n}$ from both sides of (\ref{e-diffeq}) completes the proof.
\end{proof}

\begin{lemma}
\label{l-fderiv2}
For all $n\geq0$, we have
\[
f_{n}(-z)=\frac{\D}{\D z}f_{n+1}(-z)-z\frac{\D}{\D z}f_{n}(-z)+\frac{\D}{\D z}f_{n+2}(-z).
\]
\end{lemma}

\begin{proof}
We can use the generating function
\[
\sum_{n=0}^{\infty}\frac{\D}{\D z}f_{n}(-z)x^{n}=\frac{\partial}{\partial z}G(-z,x)=\frac{-x^{2}}{(1-x+zx^{2})^{2}}
\]
from the proof of Lemma \ref{l-fderiv}.
Observe that 
\begin{align*}
\sum_{n=0}^{\infty}f_{n}(-z)x^{n} & =\frac{1}{1-x+zx^{2}}\\
 & =\frac{-x}{(1-x+zx^{2})^{2}}+\frac{zx^{2}}{(1-x+zx^{2})^{2}}+\frac{1}{(1-x+zx^{2})^{2}}\\
 & =\sum_{n=0}^{\infty}\Big(\frac{\D}{\D z}f_{n+1}(-z)-z\frac{\D}{\D z}f_{n}(-z)-\frac{\D}{\D z}f_{n+2}(-z)\Big)x^{n}.
\end{align*}
Extracting coefficients of $x^{n}$ completes the proof.
\end{proof}

\subsection{Alternative $\totalB$ formula} 
\label{subsec:genfuncs}
We are now ready to establish the connection between the generating functions $F_\ell$ and $B_\ell$ and Fibonacci polynomials, which will lead to an explicit formula for $\totalB(n,\ell )$ as a difference of two alternating sums.
\begin{theorem}
\label{t-gf}
For all $\ell\geq1$, we have \leqnomode
\[
\tag{{a}}\qquad F_{\ell }(x,y,z)=\frac{f_{\ell -1}(-xz)-f_{\ell -2}(-xz)yzC}{f_{\ell }(-xz)-f_{\ell -1}(-xz)yzC}
\]
and 
\[
\tag{{b}}\qquad B_{\ell }(z)=z^{\ell }C^{3\ell +1}\left(zC\frac{\D}{\D z}f_{\ell }(-z)-\frac{\D}{\D z}f_{\ell +1}(-z)\right).
\]
\end{theorem}
\begin{remark}
The explicit generating function in Theorem~\ref{t-gf} can also be viewed as a possible starting point for an asymptotic derivation of \(\totalB(n,\ell)\), for instance by using singularity-analysis methods for bivariate coefficients such as Drmota's theorem on coefficients of powers of generating functions~\cite[Theorem~4]{Drmota-94}. Such an approach would require singularity analysis uniform in the level parameter \(\ell\), particularly in the regime \(\ell=s\sqrt n\). In the present paper we do not pursue this second asymptotic derivation, since Theorem~\ref{thm: totalB} already gives the asymptotics needed for the BFS and DFS comparison. The purpose of this section is instead to give an independent exact-formula approach and to record the connection with continued fractions and Fibonacci polynomials.
\end{remark}
\begin{remark}
Using a basic continued fraction observation (used in \cite{de1972average}), one can guess that 
$$
F_{\ell} = \frac{T_{\ell-1}}{T_{\ell}},
$$
for some polynomial $T_{\ell}(x,y,z)$. Indeed, if $F_{\ell} = \frac{P_{\ell}}{T_{\ell}}$ for some polynomials $P$ and $T$, then Equation~\eqref{eq:F-l} implies
$$
F_{\ell} = \frac{P_{\ell}}{T_{\ell}} = \frac{1}{1-xzP_{\ell-1}/T_{\ell-1}} = \frac{T_{\ell-1}}{T_{\ell-1} - xzP_{\ell-1}},
$$
and thus $P_{\ell} = T_{\ell-1}$, so 
\begin{equation}
\label{eq:reqT}
    T_{\ell} = T_{\ell-1} - xzP_{\ell-1} = T_{\ell-1} - xzT_{\ell-2}.
\end{equation}
Now, if we know that $T_{l} = f_{\ell }(-xz)-f_{\ell -1}(-xz)yzC$, it is straightforward to verify the recurrence~\eqref{eq:reqT}, using Lemma~\ref{l-frec}. In fact, one can find an expression for $T_{l}$, and thus for $F_{\ell}$ by solving that recurrence. We give a less technical and direct inductive proof of Theorem~\ref{t-gf} a). 
\end{remark}
\begin{remark}
In order to find $F_{\ell }(x,y,z)$, one can also use the results of Kemp \cite{kemp1990number}, who found an expression for $Q_{k}(x,y) = \sum_{n\geq 1}\sum_{r\geq 1} Q_{n,k,r}x^{n}y^{r}$, where $Q_{n,k,r}$ is the number of trees in $\mathcal{T}_{n+1}$ with $r$ nodes at maximal level $k$. Specifically Lemma~2 and Equation~8 in \cite{kemp1990number}, together with the observation that $Q_{k}(x,y)$ is equal to $F_{l}(x,y,z)$ after the substitution $x\mapsto xz$, $y\mapsto \frac{y}{x}\frac{1-\sqrt{1-4z}}{2z}$, will give us an expression for $F_{l}(x,y,z)$. However, the described approach would be comparatively technical and does not exhibit the connection with the Fibonacci polynomials. 
\end{remark}

\begin{proof}[Proof of Theorem~\ref{t-gf}]
We begin by proving (a). Recall that
\[
F_{1}=\frac{1}{1-yzC} \qquad \text{and} \qquad F_{\ell }=\frac{1}{1-xzF_{\ell -1}} \quad (l \geq 2).
\]
The desired result is easily verified for $\ell=1$ using the formula for $F_1$. Proceeding via induction, let us assume that the result holds for a fixed $\ell\geq1$. Then, by the recursive formula for $F_\ell$, we have
\begin{align*}
F_{\ell +1} & =\frac{1}{1-xz\frac{f_{\ell -1}(-xz)-f_{\ell -2}(-xz)yzC}{f_{\ell }(-xz)-f_{\ell -1}(-xz)yzC}}\\
 & =\frac{1}{\frac{f_{\ell }(-xz)-f_{\ell -1}(-xz)yzC}{f_{\ell }(-xz)-f_{\ell -1}(-xz)yzC}-\frac{f_{\ell -1}(-xz)xz-f_{\ell -2}(-xz)xyz^{2}C}{f_{\ell }(-xz)-f_{\ell -1}(-xz)yzC}}\\
 & =\frac{f_{\ell }(-xz)-f_{\ell -1}(-xz)yzC}{f_{\ell }(-xz)-f_{\ell -1}(-xz)yzC-f_{\ell -1}(-xz)xz+f_{\ell -2}(-xz)xyz^{2}C}\\
 & =\frac{f_{\ell }(-xz)-f_{\ell -1}(-xz)yzC}{(f_{\ell }(-xz)-f_{\ell -1}(-xz)xz)-(f_{\ell -1}(-xz)-f_{\ell -2}(-xz)xz)yzC}\\
 & =\frac{f_{\ell }(-xz)-f_{\ell -1}(-xz)yzC}{f_{\ell +1}(-xz)-f_{\ell }(-xz)yzC},
\end{align*}
where the last equality is obtained via Lemma \ref{l-frec}. Hence (a) is proven.

We now move on to (b). In light of \eqref{e-BF}, let us proceed by taking the required derivatives. Using the formula from (a), we get 
\begin{equation}
\frac{\partial F_{\ell }}{\partial y} =\frac{(f_{\ell -1}(-xz)^{2}-f_{\ell -2}(-xz)f_{\ell }(-xz))zC}{(f_{\ell }(-xz)-f_{\ell -1}(-xz)yzC)^{2}}
 =\frac{x^{\ell -1}z^{\ell }C}{(f_{\ell }(-xz)-f_{\ell -1}(-xz)yzC)^{2}}, \label{e-dFldy}
\end{equation}
where the last step is obtained via Lemma \ref{l-lclike}. Taking the derivative with respect to $y$ and dividing by $2$ yields
\begin{equation}
\frac{1}{2}\frac{\partial^{2}F_{\ell }}{\partial y^{2}}=\frac{f_{\ell -1}(-xz)x^{\ell -1}z^{\ell +1}C^{2}}{(f_{\ell }(-xz)-f_{\ell -1}(-xz)yzC)^{3}}, \label{e-dFldy2}
\end{equation}
where
\[
\frac{1}{2}\frac{\partial^{2}F_{\ell }}{\partial y^{2}}+\frac{\partial F_{\ell }}{\partial y}=\frac{x^{\ell -1}z^{\ell }C\left(f_{\ell }(-xz)+f_{\ell -1}(-xz)(1-y)zC\right)}{(f_{\ell }(-xz)-f_{\ell -1}(-xz)yzC)^{3}}.
\]
Therefore,
\begin{equation*}
\left[\frac{1}{2}\frac{\partial^{2}F_{\ell }}{\partial y^{2}}+\frac{\partial F_{\ell }}{\partial y}\right]_{x=1,\,y=1} =\frac{f_{\ell }(-z)z^{\ell }C}{(f_{\ell }(-z)-f_{\ell -1}(-z)zC)^{3}} =z^{\ell }C^{3\ell +1}f_{\ell }(-z)
\end{equation*}
from Lemma \ref{l-denomC}.

Now, consider the derivative of $\partial F_{\ell }/\partial y$ with respect to $x$. If $\ell=1$, then $\partial F_{\ell }/\partial y$ does not contain $x$, so $\partial^{2}F_{\ell }/\partial x\partial y=0$ and therefore
\begin{align*}
B_{1} & =\left[\frac{1}{2}\frac{\partial^{2}F_{\ell }}{\partial y^{2}}+\frac{\partial F_{\ell }}{\partial y}\right]_{x=1,\,y=1}\\
 & =z^{1}C^{3(1)+1}f_{1}(-z)\\
 & =z^{1}C^{3(1)+1}\cdot1\\
 & =z^{1}C^{3(1)+1}\left(zC\frac{\D}{\D z}f_{1}(-z)-\frac{\D}{\D z}f_{2}(-z)\right).
\end{align*} 
If $\ell\geq2$, then we have
\[
\frac{\partial^{2}F_{\ell }}{\partial x\partial y}=\frac{x^{\ell -2}z^{\ell }C\left((\ell -1)\left(f_{\ell }(-xz)-f_{\ell -1}(-xz)yzC\right)-2x\left(\frac{\partial f_{\ell }(-xz)}{\partial x}-\frac{\partial f_{\ell -1}(-xz)}{\partial x}yzC\right)\right)}{(f_{\ell }(-xz)-f_{\ell -1}(-xz)yzC)^{3}}.
\]
Observing that 
\begin{equation}
\label{e-dfldx}
\left.\frac{\partial f_{\ell }(-xz)}{\partial x}\right|_{x=1}=z\frac{\D}{\D z}f_{\ell }(-z)
\end{equation}
and using Lemma \ref{l-denomC}, we get 
\begin{align}
\mkern-36mu \left.\frac{\partial^{2}F_{\ell }}{\partial x\partial y}\right|_{x=1,\,y=1} &=
\frac{z^{\ell }C\left((\ell -1)\left(f_{\ell }(-z)-f_{\ell -1}(-z)zC\right)-2\left(z\frac{\D}{\D z}f_{\ell }(-z)-z\frac{\D}{\D z}f_{\ell -1}(-z)zC\right)\right)}{(f_{\ell }(-xz)-f_{\ell -1}(-xz)yzC)^{3}}\nonumber \\
 & \mkern-50mu = z^{\ell }C^{3\ell +1}\left((\ell -1)\left(f_{\ell }(-z)-f_{\ell -1}(-z)zC\right)-2\left(z\frac{\D}{\D z}f_{\ell }(-z)-z^{2}C\frac{\D}{\D z}f_{\ell -1}(-z)\right)\right), \label{e-dFldxdy}
\end{align}
and thus 
\begin{align*}
 & B_{\ell }(z)=\left[\Big(\frac{\partial^{2}}{\partial x\partial y}+\frac{1}{2}\cdot\frac{\partial^{2}}{\partial y^{2}}+\frac{\partial}{\partial y}\Big)\,F_{\ell }\right]_{\substack{x=1,\,y=1}
}\\
 & \quad=z^{\ell }C^{3\ell +1}\left((\ell -1)\left(f_{\ell }(-z)-f_{\ell -1}(-z)zC\right)-2\left(z\frac{\D}{\D z}f_{\ell }(-z)-z^{2}C\frac{\D}{\D z}f_{\ell -1}(-z)\right)+f_{\ell }(-z)\right)\\
 & \quad=z^{\ell }C^{3\ell +1}\left(\ell f_{\ell }(-z)-2z\frac{\D}{\D z}f_{\ell }(-z)-\left((\ell -1)f_{\ell -1}(-z)-2\frac{\D}{\D z}f_{\ell -1}(-z)\right)zC\right)\\
 & \quad=z^{\ell }C^{3\ell +1}\left(zC\frac{\D}{\D z}f_{\ell }(-z)-\frac{\D}{\D z}f_{\ell +1}(-z)\right)
\end{align*}
from Lemma \ref{l-fderiv}.
\end{proof}
\begin{remark}
    As a sanity check, using Theorem~\ref{t-gf} a) and Lemma~\ref{l-denomC}, we can see that $F_{\ell}(1,1,z) = \frac{C^{-(l-1)}}{C^{-l}} = C$.
\end{remark}

By extracting coefficients from the generating function $B_{\ell } $, we get an explicit formula for $\totalB(n,\ell )$, for which we give a proof below. In this proof, we use the standard notation $[z^n]\, f$ for the coefficient of $z^n$ in $f$. 

\begin{theorem} \label{t-totalBsum}
For all $n \geq 0$ and $0 \leq \ell \leq n$, we have 
\begin{align}
\label{eq:totalBalt}
\totalB(n,\ell ) & =\sum_{k=1}^{\left\lfloor (\ell +1)/2\right\rfloor }(-1)^{k-1}k{\ell -k+1 \choose k}\frac{3\ell +1}{n-k+2\ell +2}{2n-2k+\ell +2 \choose n-\ell -k+1}\\
 & \qquad\qquad\qquad-\sum_{k=1}^{\left\lfloor \ell /2\right\rfloor }(-1)^{k-1}k{\ell -k \choose k}\frac{3\ell +2}{n-k+2\ell +2}{2n-2k+\ell +1 \choose n-\ell -k}.\nonumber
\end{align}
\end{theorem}
\begin{proof}
We first note that 
\begin{equation}
C^{m}=\sum_{n=0}^{\infty}\frac{m}{n+m}{2n+m-1 \choose n}z^{n}, \label{e-Cpow}
\end{equation}
as can be proven by a straightforward application of Lagrange inversion \cite[Equation (2.3.2)]{gessel}. Using this fact, along with Theorem \ref{t-gf} (b) and
\[
\frac{\D}{\D z}f_{\ell }(-z)=\sum_{k=0}^{\left\lfloor \ell /2\right\rfloor }(-1)^{k}k{\ell -k \choose k}z^{k-1},
\]
we get
\begin{align*}
 & \totalB(n,\ell )=[z^{n}]\,B_{\ell }(z)\\
 & \qquad=[z^{n}]\,z^{\ell }C^{3\ell +1}\left(zC\frac{\D}{\D z}f_{\ell }(-z)-\frac{\D}{\D z}f_{\ell +1}(-z)\right)\\
 & \qquad=[z^{n}]\,\left(\sum_{k=0}^{\left\lfloor \ell /2\right\rfloor }(-1)^{k}k{\ell -k \choose k}z^{\ell +k}C^{3\ell +2}-\sum_{k=0}^{\left\lfloor (\ell +1)/2\right\rfloor }(-1)^{k}k{\ell +1-k \choose k}z^{\ell +k-1}C^{3\ell +1}\right)\\
 & \qquad=[z^{n}]\,\left(\sum_{k=1}^{\left\lfloor (\ell +1)/2\right\rfloor }(-1)^{k-1}k{\ell +1-k \choose k}z^{\ell +k-1}C^{3\ell +1}-\sum_{k=1}^{\left\lfloor \ell /2\right\rfloor }(-1)^{k-1}k{\ell -k \choose k}z^{\ell +k}C^{3\ell +2}\right)\\
 & \qquad=\sum_{k=1}^{\left\lfloor (\ell +1)/2\right\rfloor }(-1)^{k-1}k{\ell +1-k \choose k}[z^{n-\ell -k+1}]C^{3\ell +1}-\sum_{k=1}^{\left\lfloor \ell /2\right\rfloor }(-1)^{k-1}k{\ell -k \choose k}[z^{n-\ell -k}]C^{3\ell +2}\\
 & \qquad=\sum_{k=1}^{\left\lfloor (\ell +1)/2\right\rfloor }(-1)^{k-1}k{\ell +1-k \choose k}[z^{n- \ell -k+1}]\sum_{n=0}^{\infty}\frac{3\ell +1}{n+3\ell +1}{2n+3l \choose n}z^{n}\\
 & \qquad\qquad\qquad\qquad\qquad-\sum_{k=1}^{\left\lfloor \ell /2\right\rfloor }(-1)^{k-1}k{\ell -k \choose k}[z^{n-\ell -k}]\sum_{n=0}^{\infty}\frac{3\ell +2}{n+3\ell +2}{2n+3\ell +1 \choose n}z^{n}\\
 & \qquad=\sum_{k=1}^{\left\lfloor (\ell +1)/2\right\rfloor }(-1)^{k-1}k{\ell -k+1 \choose k}\frac{3\ell +1}{n-k+2\ell +2}{2n-2k+\ell +2 \choose n-\ell -k+1}\\
 & \qquad\qquad\qquad\qquad\qquad-\sum_{k=1}^{\left\lfloor \ell /2\right\rfloor }(-1)^{k-1}k{\ell -k \choose k}\frac{3\ell +2}{n-k+2\ell +2}{2n-2k+\ell +1 \choose n-\ell -k}. \qedhere
\end{align*}
\end{proof}
It would be interesting to give a direct proof that our two expressions for $\totalB(n,\ell )$---in Equations~\eqref{eq:totalBExpr} and \eqref{eq:totalBalt}---are in fact equal. Recovering the asymptotics of $\totalB(n,\ell)$ obtained in Theorem~\ref{thm: totalB} from Equation~\eqref{eq:totalBalt} would be also interesting and we were not able to do that.

\section{General results on the asymptotics of totalB for Galton--Watson trees}\label{sec: CRT}

A well-known class of random rooted trees are the \textit{Galton--Watson
trees}. In this section, we generalize Theorem~\ref{thm: totalB} to Galton--Watson trees. Here, we present a short overview of these objects. For a thorough and rigorous presentation of Galton--Watson trees, and
random trees in general, we refer to the surveys by Le Gall \cite{LeGall-05} and Janson \cite{Janson-12}, as well as the book of Drmota \cite{Drmota-09}.

Let $Z$ be a random variable on $\{0,1,2,\dots\}$ with
\[
\Ebb{Z} = 1
\quad
\text{and}
\quad
\Var{Z} = \sigma^2.
\]
We define a Galton--Watson tree $T^{(Z)}$ with an offspring
distribution $Z$ recursively, by giving each node independently a
random number of children distributed as $Z$. Let us define
$\mathcal{T}^{(Z)}_n$ to be
the set of all possible trees $T^{(Z)}$ with $n$ edges, endowed
with the conditional probability induced by $Z$. As observed by Aldous (see \cite[p. 28]{CRT2}
and \cite[p.12]{Drmota-09}),
the uniform distribution on some frequently used sets
of rooted trees is generated by $\mathcal{T}_n^{(Z)}$ for a suitable $Z$. For 
example:
\begin{itemize}[leftmargin=25bp] \setlength\itemsep{5bp}
    \item ordered trees, when $Z$ has a geometric
    distribution---i.e., $\P(Z=k) = 1/2^{k+1}$ for $k \geq 0$.
    \item strict $m$-ary trees (every node has $0$ or $m$ children), when $\P(Z=0) = \P(Z=m) = 1/2$.
    \item $m$-ary trees when $Z$ has a binomial $\operatorname{Bin}(m,1/m)$ distribution.
    \item labeled trees when $Z$ has a Poisson distribution with parameter 1---i.e., when
    $\P(Z = k) = e^{-1}/k!$.
\end{itemize}
We obtain
the asymptotics of $\totalB(n,s\sqrt{n})$ on
all these sets of trees; see Corollary \ref{cor: totalB GW asymptotics}.

As we have already seen in Section \ref{sec:totalD}, an ordered tree with $n$ edges can be
associated with a discrete excursion from $0$
to $2n$ via its corresponding Dyck path. It is therefore
natural to expect that, after scaling, the
corresponding Dyck path converges weakly
to a standard Brownian excursion
$\mathbf{e}$ of length 1, and thus
the trees themselves will have a weak limit in some
abstract sense. This was formalized in
the classical works of Aldous \cite{CRT1, CRT2, CRT3}, who called this limiting object
the \textit{continuum random tree}. 

Using the random variables defined in Equation~\eqref{def:Random variables}, in order to find $\totalB(n,l)$ in the general case, we can use the equation
\begin{align}
    S_n(\ell ) = \binom{H_n(\ell )}{2} - \binom{H_n(\ell -1)}{2} &= \frac{1}{2}\lbrb{H_n(\ell ) - H_n(\ell -1)}
    \lbrb{H_n(\ell ) + H_n(\ell -1) - 1} \nonumber \\
    &= \frac12 h_n(\ell )\lbrb{2H_n(\ell ) + h_n(\ell )-1}. \label{eq: S_n(l)}
    \end{align}
Following the work of Aldous, the random variables $h_n(\ell )$ and $H_{n}(\ell )$---giving the number of vertices at level $\ell$
and the accumulated number of vertices up to level $\ell$, respectively---should converge to the respective quantities for the Brownian excursion. To formalize this, first extend
$h_n$ and $H_n$ by linear interpolation to be defined on all reals in $[0,2n]$ instead of only on integers. That is, define
\[
h_n(t) \coloneqq \lbrb{\lfloor t\rfloor+1-t}h_n\lbrb{\lfloor t\rfloor} + \lbrb{t-\lfloor t\rfloor}h_n\lbrb{\lfloor t\rfloor+1}
\]
for $t \in [0,2n]$, and $h_n(t) \coloneqq 0$ for $t\notin [0,2n]$. Next, let us define
\[
\mu_s \coloneqq \int_0^1 \ind{\mathbf{e}_t \leq s}\, \D s
\quad
\text{and}
\quad
L_s \coloneqq \frac{\D}{\D s}\mu_s = \lim_{\epsilon \to 0}\frac{1}{\epsilon}
\int_0^1 \ind{\mathbf{e}_t \in [s, s+\epsilon]}\,\D t,
\]
called, respectively, the occupation measure
and occupation density (or local time) of $\mathbf{e}$. It was proven by Aldous in
\cite[Corollary 3]{CRT2} that for a tree
in $\mathcal{T}_n^{(Z)}$,
\[
\lbrb{\frac{1}{n}H_n\lbrb{s\sqrt{n}}}_{s\geq 0} \xrightarrow[\quad]{d} \lbrb{\mu_{\sigma s/2}}_{s\geq 0},
\]
as processes on the set of continuous real
functions on $[0,1]$, endowed with the uniform norm.
It was also conjectured therein (see \cite[Conjecture 
4]{CRT2}), and proven by Drmota and Gittenberger
in \cite{Drmota-Gettenberger-97}, that
\begin{equation}
    \label{eq: h_n(l) limit}
\lbrb{\frac{1}{\sqrt{n}}h_n\lbrb{s\sqrt{n}}}_{s\geq 0} \xrightarrow[\quad]{d} \lbrb{\frac{\sigma}{2}L_{\sigma s/2}}_{ s\geq 0}.
\end{equation}
We now have the needed background to state
the main result of this section.

\begin{theorem}
\label{th:generaltotalB}
For random trees from $\mathcal{T}_n^{(Z)}$, we have
 \[
\frac{1}{n^{3/2}}\Ebb{S_n\lbrb{s\sqrt{n}}} \to \sigma K_{\sigma s},
\]
for any constant $s \geq 0$ as $n \to \infty$, where
\begin{equation}
    \label{eq:def K_s}
K_s \coloneqq \frac{1}{2} \Ebb{ L_{s/2} \mu_{s/2}}=
2se^{-s^2} - 2\sqrt{\pi}\lbrb{\Phi\lbrb{2\sqrt{2}s} - \Phi\lbrb{\sqrt{2}s}}.
\end{equation}
 As a consequence, if the trees in $\mathcal{T}_n^{(Z)}$ are equiprobable, then we have
 \begin{equation}
 \label{eq:totalB-GW-asymp}
\totalB\lbrb{n,s\sqrt{n}} \sim \sigma K_{\sigma s} n^{3/2}\left|\mathcal{T}_n^{(Z)}\right|     
 \end{equation}
as $n\to \infty$.
\end{theorem}

\begin{proof}
From \eqref{eq: S_n(l)}, we know that
    \begin{equation}\label{eq: S_n in proof}
        \begin{split}
   \frac{1}{n^{3/2}} \Ebb{S_n\lbrb{s\sqrt{n}}} =  \Ebb{ \frac{h_n\lbrb{s\sqrt{n}}}{\sqrt{n}}\cdot \frac{{H_n\lbrb{s\sqrt{n}}}}{n}} + \frac{1}{2n^{3/2}}  \lbrb{\Ebb{h_n^2\lbrb{s\sqrt{n}}} -\Ebb{ h_n\lbrb{s\sqrt{n}}}}.
        \end{split}
    \end{equation}
    By \cite[Theorem 3]{Drmota-Gittenberger-04}, the last two expectations in \eqref{eq: S_n in proof} are of order $n$ and $n^{1/2}$, respectively, so we need only consider the first term. By a standard argument---see, e.g., \cite[Theorem 3]{Aldous-98}---Equation~\eqref{eq: h_n(l) limit} implies
    \[
    \lbrb{\frac{h_n\lbrb{s\sqrt{n}}}{\sqrt{n}} , \frac{H_n\lbrb{s\sqrt{n}}}{n}}_{s \geq 0} \xrightarrow[\quad]{d} \lbrb{\frac{\sigma}{2}L_{\sigma s/2},\mu_{\sigma s/2}}_{s \geq 0}.
    \]
    Indeed, by Skorohod--Dudley's theorem \cite[Theorem 5.31]{Kallenberg-21}, we can find a coupling such that the convergence in \eqref{eq: h_n(l) limit} is almost sure and therefore
    \[
    \lbrb{\frac{h_n\lbrb{s\sqrt{n}}}{\sqrt{n}}, \int_0^s \frac{h_n(y\sqrt{n})}{\sqrt{n}}\, \D y}_{s \geq 0} \xrightarrow[\quad]{d} \lbrb{\frac{\sigma}{2}L_{\sigma s/2},\frac{\sigma}{2} \int_0^s L_{\sigma y/2}\, \D y }_{s \geq 0}.
    \]
    The last term is $H_{\sigma s/2}$ by the definition of the local time $\ell$, and it remains
    to link $H_n$ with the integral of $h_n$, which is readily done. As an illustration, consider
    the case where
    $s\sqrt{n}$ is an integer. Then, we have
    \[
    \frac{H_n\lbrb{s\sqrt{n}}}{n} = \frac{1}{n}\lbrb{\int_0^{s\sqrt{n}} {h_n(y)}\, \D y + \frac12 \lbrb{1 + h_n\lbrb{s\sqrt{n}}}} = \int_0^s \frac{h_n\lbrb{y\sqrt{n}}}{\sqrt{n}}\, \D y + \so{1}.
    \]
    Therefore, by the continuous mapping theorem, we have
    \begin{equation}\label{eq: weak hH}
    \frac{1}{n^{3/2}} h_n\lbrb{s\sqrt{n}}H_n\lbrb{s\sqrt{n}} \xrightarrow[\quad]{d} 
   \frac{\sigma}{2} L_{\sigma s/2}\mu_{\sigma s/2}.
    \end{equation}
    To obtain convergence of the means from the last convergence in distribution, we need to check 
    for uniform integrability. For this, it will be enough to ensure an uniform $L^p$ bound for some $p>1$; see, e.g., \cite[Theorem T22]{Meyer-66}. This is indeed the case as
    \[
      H_n(s \sqrt{n}) \leq \sum_{j \leq s\sqrt{n} + 1} h_n(j) \leq  (s\sqrt{n} + 1) \sup_m h_n(m),      
    \]
    so
    \[
     \frac{1}{n^{3/2}} h_n\lbrb{s\sqrt{n}}H_n\lbrb{s\sqrt{n}} \leq \frac{(s\sqrt{n} + 1)}{\sqrt{n}} \lbrb{
     \frac{\sup_m h_n(m)}{\sqrt{n}}
     }^2,
    \]
    and the last quantity is uniformly bounded in $L^p$ by \cite[Theorem 3]{Drmota-Gittenberger-04}.
    Therefore, \eqref{eq: weak hH} implies that
    \[
        \frac{1}{n^{3/2}}\Ebb{h_n\lbrb{s\sqrt{n}}{H_n\lbrb{s\sqrt{n}}}} \xrightarrow[n \to \infty]{} \sigma \Ebb{\frac{1}{2}
        L_{\sigma s/2} \mu_{\sigma s/2}}.
    \]
    Substituting into \eqref{eq: S_n in proof},
    we get
    \begin{equation} \label{e-con1}
    \frac{1}{n^{3/2}}\Ebb{S_n(s\sqrt{n})}  \to \sigma K_{\sigma s}.
    \end{equation}
    It is immediate from \eqref{e-con1} that if the probability on $\mathcal{T}_n^{(Z)}$ is uniform, then
    \begin{equation} \label{e-con2}
    \frac{1}{n^{3/2}} \cdot \frac{\totalB\lbrb{n, s\sqrt{n}}}{|\mathcal{T}_n^{(Z)}|} \to \sigma
    K_{\sigma s}.
    \end{equation}
    It remains to identify the constant $K_{s} = \Ebb{L_{s/2}\mu_{s/2}}/2$. Recall that random ordered trees can be constructed as Galton--Watson trees with offspring distribution $Z$, such that $\P(Z=k) = 1/2^{k+ 1}$ and $\Var Z = 1$. Therefore \eqref{e-con2} reads
    \[
    \totalB\lbrb{n, s\sqrt{n}} \xrightarrow[n \to \infty]{} n^{3/2} \left|\mathcal{T}_n\right| K_s.
    \]
    However, we know that $|\mathcal{T}_n|$ is equal to the $n$th Catalan number $C_n$, and 
    using $C_{n} = \frac{4^{n}}{\sqrt{\pi}n^{3/2}}(1+\Theta(\frac{1}{n}))$, we obtain
    \[
     \totalB\lbrb{n, s\sqrt{n}} \sim \frac{K_s}{\sqrt{\pi}}4^n.
    \]
    Comparing this with Equation~\eqref{eq:asymptotics totalB} in Theorem \ref{thm: totalB}, we find the desired 
    explicit expression for $K_s$. 
\end{proof}
\begin{corollary}\label{cor: totalB GW asymptotics}Let $K_s$ be defined as in Equation~\eqref{eq:def K_s}, and let $s \geq 0$.
    \begin{enumerate}
        \item[\textup{(a)}] On $m$-ary trees---i.e., when $Z = \operatorname{Bin}(m, 1/m)$
        and $\sigma = \sqrt{(m-1)/m}$---we have
        \[
        \totalB\lbrb{n,s\sqrt{n}} \sim \sigma K_{\sigma s} \frac{n^{3/2} }{(m-1)n + 1}\binom{mn}{n}
        \sim  \frac{K_{\sigma s}}{\sqrt{2\pi}}\frac{m^{mn}}{(m-1)^{(m-1)n+1}}.
        \]
        \item[\textup{(b)}] On rooted labeled trees---i.e., when $Z\sim \operatorname{Poi}(1)$ and $\sigma = 1$---we have
        \[
        \totalB\lbrb{n,s\sqrt{n}} \sim K_{s} n^{3/2} n^{n-1} = K_s n^{n+1/2}.
        \]
    \end{enumerate}
\end{corollary}
\section{Search in case of a known target level}
\label{sec:truncDFS}
In Section~\ref{sec:threshold}, we compared the average performance of BFS and DFS when searching for a random node at a fixed level $\ell$, and we determined which of the two algorithms has an advantage depending on $\ell$. However, if we know $\ell$ in advance, it is reasonable to use a simple modification of the DFS algorithm that we call \emph{truncated DFS}, which does not visit nodes at levels larger than $\ell$.
\begin{definition}[Truncated DFS]
Apply DFS and ignore nodes at levels $>\ell$, where $\ell$ is fixed.
\end{definition}
Let us define
$$
\dfsTruncScore(v)\coloneqq \text{the number of nodes visited before $v$ when using truncated DFS.}
$$
Figure~\ref{fig:dfs_trunc} gives an example of a tree traversal when using truncated DFS with $\ell=2$. Each node is labeled with its $\dfsTruncScore$. Note that nodes at levels above $\ell=2$ are not visited and that $\dfsTruncScore(v)$ gives the number of steps made by the algorithm if the target is $v$.
\begin{figure}[H]
  \centering
      \begin{tikzpicture}  
  [baseline=0, scale=0.8,auto=center]
    
  \node[circle, draw, fill=black!100, inner sep=0pt, minimum width=4pt] (a1) at (0,10) {};
  \node[above] at (0,10) {0}; 
  \node[circle, draw, fill=black!100, inner sep=0pt, minimum width=4pt] (a2) at (-1.5,9) {};
  \node[left] at (-1.5,9) {1};
  \node[circle, draw, fill=black!100, inner sep=0pt, minimum width=4pt] (a3) at (0,9) {};
  \node[left] at (0,9) {4};
  \node[circle, draw, fill=black!100, inner sep=0pt, minimum width=4pt] (a4) at (1.5,9) {};
  \node[left] at (1.5,9) {5};
  \node[circle, draw, fill=black!100, inner sep=0pt, minimum width=4pt] (a5) at (-2.5,8) {}; 
  \node[left] at (-2.5,8) {2};
  \node[circle, draw, fill=black!100, inner sep=0pt, minimum width=4pt] (a6) at (-0.5,8) {}; 
  \node[left] at (-0.5,8) {3};
  \node[circle, draw, fill=black!100, inner sep=0pt, minimum width=4pt] (a7) at (-2.5,7) {}; 
  \node[circle, draw, fill=black!100, inner sep=0pt, minimum width=4pt] (a8) at (1.5,8) {};
  \node[left] at (1.5,8) {6};
  \node[circle, draw, fill=black!100, inner sep=0pt, minimum width=4pt] (a9) at (0.5,7) {};
  \node[circle, draw, fill=black!100, inner sep=0pt, minimum width=4pt] (a10) at (2.5,7) {};
  
  \draw (a1) -- (a2);
  \draw (a1) -- (a3);  
  \draw (a1) -- (a4);  
  \draw (a2) -- (a5);  
  \draw (a2) -- (a6);
  \draw (a5) -- (a7);  
  \draw (a4) -- (a8);  
  \draw (a8) -- (a9);  
  \draw (a8) -- (a10);  
\end{tikzpicture}
\vspace{-5cm}
  \caption{The order of exploration for the truncated DFS algorithm.}
  \label{fig:dfs_trunc}
\end{figure}
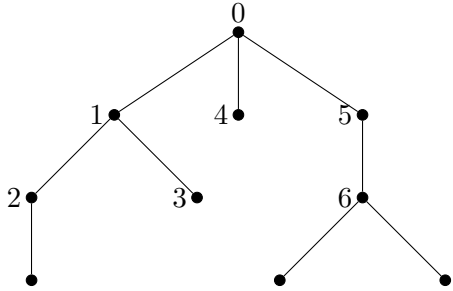
One can see that the truncated DFS has a smaller expected number of steps compared to both BFS and DFS. For DFS, this is immediate, since while using DFS we visit additional nodes that are skipped when using the truncated version. For BFS, observe that in any tree the nodes at level $\ell$ have the largest bfsScores among all nodes at levels below $\ell+1$. However, for the truncated DFS, for most of the trees, the nodes at level $\ell$ are not those with the largest $\dfsTruncScore$. Thus the truncated DFS is also faster than BFS on average.

Our next goal will be to prove a statement on $\totalDTrunc(n,\ell )$ exploiting a particular symmetry, which will be used later to find the asymptotic average time complexity of the truncated DFS algorithm. Recall from Section~\ref{sec:fibo} that 
\[
F_{\ell }=\sum_{n=0}^{\infty}\sum_{T\in{\mathcal T}_n}x^{k(T)}y^{m(T)}z^{n},
\]
where $k=k(T)$ is the number of non-root nodes in $T$ at levels smaller than $\ell$ and $m=m(T)$ is the number of nodes in $T$ at level $\ell$.
\begin{theorem}
\label{t-truncdiff}
For all $n \geq 1$ and $0 \leq \ell \leq n$, we have
$$\totalDTrunc(n,\ell ) = \totalD(n,\ell ) - \frac{1}{2} \sum_{T \in \mathcal{T}_n} (m(T)-1)(n-k(T)-m(T)).$$
\end{theorem}
\begin{proof}
We will show that $\totalDTrunc(n,\ell )$ is equal to $\totalD(n,\ell )$ minus the number of pairs of nodes $(x,y)$ among all trees $T \in \mathcal{T}_n$ satisfying the following properties:
\begin{enumerate} \setlength\itemsep{2bp}
\item[(a)] $x$ is at level $\ell$;
\item[(b)] $y$ is at level larger than $\ell$;
\item[(c)] $y$ is visited before $x$ in the inorder traversal of $T$;
\item[(d)] $y$ is not a descendant of $x$.
\end{enumerate}
To see this, fix a node $x$ at level $\ell$ in a tree $T \in \mathcal{T}_n$. Then we can obtain $\dfsTruncScore(x)$ from $\dfsScore(x)$ by subtracting the number of nodes $y$ in $T$ with the above properties, since these are precisely the nodes that are not visited when using the truncated DFS algorithm to search for $x$, but are visited when using (ordinary) DFS.

It remains to count the pairs $(x,y)$ satisfying the above properties. Suppose that $T \in \mathcal{T}_n$ has $k$ nodes at levels smaller than $\ell$ and $m$ nodes at level $\ell$. If we only care about the levels of $x$ and $y$, i.e., if we disregard conditions (c) and (d) above, then there are $m(n-k-m)$ pairs in $T$. Out of those, there are $n-k-m$ pairs for which $y$ is a descendant of $x$. For each of the $n-k-m$ possible nodes $y$, there is exactly one node $x$ at level $\ell$ of which $y$ is a descendant. Therefore, $m(n-k-m) - (n-k-m) = (m-1)(n-k-m)$ pairs remain. We must sum up over all trees in $T_n$ and divide the result by 2, because by symmetry $y$ is visited before $x$ when performing inorder traversal in exactly half of these pairs.
\end{proof}
The latter theorem implies that if we define the generating function $D_{\ell }$ by
\begin{align}
D_{\ell }=D_{\ell }(z) & \coloneqq\sum_{n=0}^{\infty} \sum_{T \in \mathcal{T}_n} (m(T)-1)(n-k(T)-m(T))z^{n} \nonumber \\
 & =\left[\Bigg(z\frac{\partial^{2}}{\partial yz}-\frac{\partial^{2}}{\partial xy}-z\frac{\partial}{\partial z}+\frac{\partial}{\partial x}-\frac{\partial^{2}}{\partial y^{2}}\Bigg)F_{\ell }\right]_{x=1,\,y=1}, \label{e-Dl}
\end{align}
the coefficient of $z^{n}$ in $D_{\ell }$ will be twice the difference between $\totalD(n,\ell )$ and $\totalDTrunc(n,\ell )$. One can use this expression for $D_\ell$ to obtain an alternating summation formula for $\totalDTrunc(n,\ell )$, as the one in Equation~\eqref{eq:totalBalt} for $\totalB(n,\ell)$. This can be done again by a similar generating function approach to the one used in Section~\ref{subsec:genfuncs}. We will not include this explicit $\totalDTrunc(n,\ell )$ formula here, since the obtained expression is comprised of several complicated alternating sums, and the generating function approach is very technical. Instead, we will prove that the truncated DFS algorithm is twice as fast as BFS, in the average case, for any given level $\ell = s\sqrt{n}$, where $s$ is a constant. 
\subsection{Asymptotic complexity of the truncated DFS}
\label{sec:complDFStrunc}
Here, we find an expression (Equation\eqref{eq:exact}) for the complexity of the truncated DFS algorithm, $totalDTrunc(n,\ell)$, for any target level $\ell$. We also evaluate this expression when $s$ is a constant. Recall from Section~\ref{sec:totalB}, that $h_n(\ell)$, $H_n(\ell)$ and $\nu_n(\ell)$ are denoting the number of vertices at levels $\ell$, less or equal to $\ell$ and greater or equal to $\ell$, respectively.  If $N(n,\ell) := C_n\,\mathbb E[h_n(\ell)]$ denotes the exact total number of nodes at level $\ell$ among trees in $T_n$, then Theorem~\ref{t-dersh} gives
\[
N(n,\ell) := C_n\,\mathbb E[h_n(\ell)] = \frac{2\ell+1}{2n+1}\binom{2n+1}{n-\ell},
\]
Recall from Section~\ref{subsec:ratios} that
\begin{equation}
\frac{\totalB(n,s\sqrt n)}{4^n} \;\longrightarrow\; B(s) := \frac{2s}{\sqrt\pi}e^{-s^2} - 2\Delta(s), \qquad n\to\infty,
\label{eq:Bform}
\end{equation}
where
\begin{equation}
\Delta(s) := \Phi(2\sqrt2\,s) - \Phi(\sqrt2\,s).
\end{equation}
Recall also that
\begin{equation}
\frac{\totalD(n,s\sqrt n)}{4^n}\longrightarrow g(s):=\frac{s}{\sqrt\pi}e^{-s^2}.
\end{equation}
By Theorem~\ref{t-truncdiff},
\begin{align}
\totalD(n,\ell) - \totalDTrunc(n,\ell)
&= \tfrac12\sum_{T\in T_n}\bigl(m(T)-1\bigr)\bigl(n-k(T)-m(T)\bigr) \notag\\[4pt]
&= \tfrac12\,C_n\,\mathbb{E}\bigl[(h_n(\ell)-1)\,\nu_n(\ell+1)\bigr],
\label{eq:thm72}
\end{align}
Here, we are using that $n-k(T)-m(T)=\nu_n(\ell+1)=n+1-H_n(\ell)$, which follows directly from the definitions of $k(T)$ and $m(T)$ given before Theorem~\ref{t-truncdiff}.
 
For a tree with $m=h_n(\ell)$ nodes at level $\ell$, visited in BFS order, right after all of the $H_n(\ell-1)$ nodes at levels less than $\ell$, we get $H_n(\ell-1), H_n(\ell-1)+1,\dots,H_n(\ell-1)+m-1$ for the BFS scores of the nodes on level $\ell$. Summing up and using that $h_{n} = m$, we obtain
\[
S_n(\ell) = m\,H_n(\ell-1) + \binom m2 = m\bigl(H_n(\ell)-m\bigr) + \binom m2 = h_n(\ell)H_n(\ell) - \tfrac{1}{2}h_n(\ell)\bigl(h_n(\ell)+1\bigr).
\]
Since $C_n\mathbb E[S_n(\ell)] = \totalB(n,\ell)$ by definition, then taking the expectation on both sides, multiplying by $C_n$ and reordering leads to
\begin{equation}
C_n\,\mathbb E[h_n(\ell)H_n(\ell)] \;=\; \totalB(n,\ell) + \tfrac12\,C_n\,\mathbb E[h_n(\ell)^2] + \tfrac12\,N(n,\ell).
\label{eq:hH}
\end{equation}
Since $\nu_n(\ell + 1) = n+1 - H_n(\ell)$, we expand the product in \eqref{eq:thm72} to get
\[
(h_n(\ell)-1)(n+1-H_n(\ell)) = (n+1)h_n(\ell) - h_n(\ell)H_n(\ell) - (n+1) + H_n(\ell).
\]
Taking, again, the expectation on both sides, multiplying by $C_n$, and using Equation~\eqref{eq:thm72} and Equation~\eqref{eq:hH}, gives us 
\begin{align}
\totalD(n,\ell)-\totalDTrunc(n,\ell) &= \tfrac12\Bigl[(n+1)N(n,\ell) - \totalB(n,\ell)\Bigr]
- \tfrac14 C_n\mathbb E[h_n(\ell)^2] - \tfrac14 N(n,\ell)  \notag\\ &-\tfrac12(n+1)C_n\;+\;\tfrac12 C_n\mathbb E[H_n(\ell)].
\label{eq:exact}
\end{align}
Let $s>0$ be a fixed constant and let $\ell=s\sqrt n$, $n\to\infty$. As pointed out in the proof of Theorem~\ref{th:generaltotalB}, in this case we have 
$E[h_n(s\sqrt n)] = \mathcal{O}(\sqrt{n})$ and 
$E[h_n(s\sqrt n)^{2}] = \mathcal{O}(n)$. Thus,
\[
N(n,\ell) = C_n\,\mathbb E[h_n(s\sqrt n)] = \mathcal{O}(\sqrt n\,C_n)=\mathcal{O}\!\left(\tfrac{4^n}{n}\right),\qquad
C_n\,\mathbb E[h_n(s\sqrt n)^2] = \mathcal{O}(n\,C_n) = \mathcal{O}\!\left(\tfrac{4^n}{\sqrt n}\right),
\]
and trivially $C_n\mathbb E[H_n(\ell)]\le (n+1)C_n = \mathcal{O}(4^n/\sqrt n)$. To see the latter, recall that $C_n \sim \frac{4^{n}}{\sqrt{\pi}n^{3/2}}$. Hence every term in \eqref{eq:exact} except the first one is $o(4^n)$.
 
For that first term, the known formulas for $N(n,\ell)$ and $\totalD(n,\ell)$ given by Equation~\eqref{eq:N} and Equation~\eqref{eq:totalD}, respectively, give
\begin{equation}
n\cdot\frac{N(n,\ell)}{\totalD(n,\ell)}  = n\frac{\frac{2\ell+1}{2n+1}\binom{2n+1}{n-\ell}}{\ell\binom{2n}{n-\ell}} =  \frac{n(2\ell+1)}{\ell(n+\ell+1)}\;\xrightarrow[\ell=o(n)]{}\;2,
\label{eq:nN}
\end{equation}
so that $(n+1)N(n,\ell)\sim 2\,\totalD(n,\ell)$ for $\ell=s\sqrt n$. Thus, dividing \eqref{eq:exact} by $4^n$ and taking $n\to\infty$ gives
\begin{equation}
\frac{\totalD(n,\ell)-\totalDTrunc(n,\ell)}{4^n}\;\longrightarrow\; g(s) - \tfrac12 B(s).
\label{eq:diff}
\end{equation}
Combining \eqref{eq:diff} with \eqref{eq:Bform} results in
\[
g(s)-\tfrac12 B(s) = g(s) - \Bigl(g(s)-\Delta(s)\Bigr) = \Delta(s),
\]
so $\totalD(n,\ell)-\totalDTrunc(n,\ell) \sim 4^n\Delta(s)$, and therefore we proved the last major result of this paper, formulated below.
\begin{theorem}
\label{th:dTruncBFSratio}
Let $\ell = s\sqrt n$ for a fixed constant $s>0$. As $n\to\infty$,
\begin{equation}
\frac{\totalDTrunc(n,\ell)}{4^n} \;\longrightarrow\; g(s)-\Delta(s) \;=\; \frac{s}{\sqrt\pi}e^{-s^2} \;-\;\Phi(2\sqrt2\,s)+\Phi(\sqrt2\,s) \;=\; \frac{B(s)}{2}.
\label{eq:main}
\end{equation}
Equivalently, when $s>0$ is a fixed constant and $n\to\infty$, we have
\begin{equation}
\label{eq:totalDTruncTotalB}
\totalDTrunc(n,s\sqrt n)\;\sim\;\tfrac12\,\totalB(n,s\sqrt n).
\end{equation}
\end{theorem}
\begin{remark}
\label{rem:main}
    We believe that Equation~\eqref{eq:totalDTruncTotalB} holds for an arbitrary level $\ell = s\sqrt{n}$, but not just when $s$ is a constant. The most immediate approach would be to use again Equation~\eqref{eq:exact}, yet we were not able to extend our statement to the cases of $s$ being a function of $n$. The heuristic argument sketched on Figure~\ref{fig:heuristic} below provides some additional intuition why the equation may hold in general (see the caption). 
\end{remark}
\begin{figure}[h!]
        \centering
\begin{tikzpicture}[>=stealth, scale=0.5]
  \coordinate (root) at (0,4);
  \coordinate (left) at (-3,0);
  \coordinate (right) at (3,0);
  \draw[thick] (root) -- (left) -- (right) -- cycle;
  \filldraw (root) circle (2pt);
  \node[above] at (root) {root};
  \coordinate (Lleft)  at (-1.8,1.6);
  \coordinate (Lright) at ( 1.8,1.6);
  \draw[dashed] (Lleft) -- (Lright);
  \node[left] at (Lleft) {level $\ell$};
  \coordinate (v) at (0,1.6);
  \filldraw[red] (v) circle (2.2pt);
  \node[above right] at (v) {$v$};

  \fill[gray!25] (0,4) -- (-1.8,1.6) -- (0,1.6) -- cycle;
\end{tikzpicture}
        \caption{Given that the target node $v$ is in a fixed tree, on average $v$ would be in the middle of level $\ell$. The nodes in the gray area (traversed with the truncated DFS) are roughly $1/2$ of those traversed by BFS (the nodes above the dashed line denoting level $\ell$ or on that level, but visited prior to $v$).}
        \label{fig:heuristic}
    \end{figure}
    
\section{Further questions}
\label{sec:q}

To conclude this paper, we present the following questions for future study.

\begin{question}
Suppose that we compare BFS and DFS for a random target node at a given level, but in binary trees. Is there, again, a unique threshold where DFS becomes faster than BFS, which is close to the average level for a node in binary trees? If yes, is there a small difference between the threshold and this average node level, as proven for ordered trees at the end of Section~\ref{sec:threshold}? \label{q-binary}
\end{question}

The results in Section~\ref{sec: CRT} give asymptotics for BFS on general conditioned Galton--Watson trees. It would be interesting to develop a corresponding asymptotic theory for DFS-type quantities, including ordinary DFS and truncated DFS. One possible analytic approach would be to derive functional equations for the generating functions encoding the relevant score quantities in the framework of simply generated trees; see e.g. Drmota~\cite[Section 1.2.7]{Drmota-09}. Such a method could potentially yield asymptotic results not only for ordered trees, but also for broader classes of simply generated trees, equivalently for associated conditioned Galton--Watson trees.

There may also be a probabilistic route through the continuum random tree. For BFS-type quantities, profile and local-time convergence are sufficient, since BFS depends essentially on the number of vertices at prescribed levels and on cumulative profiles. In contrast, DFS-type quantities are sensitive to the depth-first order of the vertices. A CRT approach to ordinary DFS or truncated DFS would therefore seem to require convergence of more refined functionals of the discrete contour process, such as time-indexed measures of level upcrossings. We are not aware of such a convergence result in a directly applicable form, but it seems to be a natural direction for future work. We ask the same question for the truncated DFS algorithm, as well.

\begin{question}

Can one obtain asymptotic formulas for \(\totalD(n,\ell)\) and \(\totalDTrunc(n,\ell)\), uniformly in regimes such as \(\ell=s\sqrt n\), for general simply generated trees or conditioned Galton--Watson trees?

\end{question}

The next question is related to Section~\ref{sec:truncDFS}. In the standard implementation of BFS, we push the root into a queue. Then, we push the children of the node at the top of the queue, until we traverse all of them. Then, we pop that node and continue in the same way with the new element that is at the top. DFS is executed in the same way, but the nodes are pushed onto a stack instead of a queue. It is sensible to allow pushing nodes on both sides of a double-ended queue (\emph{deque}). 
\begin{question}
\label{q:deq}
What is the optimal search algorithm implemented using a deque, if $\ell$ is known in advance and the target is selected uniformly at random among all nodes on level $\ell$ in $\mathcal{T}_{n}$?    
\end{question}
We believe that truncated DFS is not only faster in expectation than both BFS and DFS, but that it is faster in expectation than any algorithm that can be implemented with a deque.
\begin{question}
Given that $n$ and the target level $\ell$ are known in advance, is it true that the truncated DFS algorithm defined in Section~\ref{sec:truncDFS} is faster in expectation than any algorithm that uses only a deque, i.e., an algorithm that can switch between BFS and DFS at every step? \label{q-deque}
\end{question}

\begin{question}
Suppose that we do not know the exact level $\ell$ where our target node is located, but only that $\ell$ is in an interval $[u,v]$. What is the optimal algorithm, using the deque data structure, in this setting?
\end{question}

\begin{question}
Suppose that we have multiple target nodes chosen uniformly at random among those at a given level $\ell$ in the trees with $n$ edges. What are the average-case time complexities of BFS and DFS in this scenario?
\end{question}

\section*{Acknowledgments}
We wish to recognize and thank Luz Grisales G\'{o}mez for her involvement at the beginning stages of this project. The first author is thankful to Svante Janson for pointing out reference \cite{CRT2}.
The same author is also indebted to George Spahn for mentioning the idea of the truncated DFS algorithm. The first author was partially supported by an AMS-Simons Travel Grant. The second author is supported by the European Union's NextGenerationEU, through the National Recovery and Resilience Plan of the Republic of Bulgaria, project No BG-RRP-2.004-0008. The third author was partially supported by an AMS-Simons Travel Grant and NSF grant DMS-2316181. Finally, Claude Sonnet 5 was used to obtain some of the results in Section 4.1 and in Section 7.1, as well as for catching mistakes throughout the paper.

\bibliographystyle{plain}
\bibliography{bibliography}

\end{document}